\documentclass[
oneside,
reqno,
10pt
]
{amsart}
\usepackage[utf8]{inputenc}
\usepackage{amsmath}
\usepackage{amsfonts}
\usepackage{amssymb}
\usepackage[title]{appendix}
\usepackage{amsthm}
\usepackage[a4paper]{geometry}
\usepackage{color}
\usepackage{latexsym}
\usepackage[colorlinks=true,linkcolor=blue,citecolor=mauve]{hyperref}
\definecolor{mauve}{rgb}{0.58,0,0.82}

\usepackage{esint}

\let\phi\varphi
\let\epsilon\varepsilon
\newtheorem{thm}{Theorem}[section]

\newtheorem{lem}[thm]{Lemma}
\newtheorem{cor}[thm]{Corollary}
\theoremstyle{definition}
\newtheorem{df}[thm]{Definition}

\newtheorem*{rem}{Remark}

\numberwithin{equation}{section}

\author{S. Barth \and A. Bitter}
\thanks{This work was supported by the
Deutsche Forschungsgemeinschaft (DFG) through the Research Training
Group 1838: Spectral Theory and Dynamics of Quantum Systems}
\address{Department of Mathematics, University of Stuttgart, 70569 Stuttgart, Germany}
\email{simon.barth@mathematik.uni-stuttgart.de}
\email{andreas.bitter@mathematik.uni-stuttgart.de}
\date{}
\title[Two-body interactions and three-body systems]{On the virtual level of two-body interactions and  applications to three-body systems in higher dimensions}
\begin{document}
 
\begin{abstract}
We consider a system of three particles in dimension $d\geq 4$ interacting via short-range potentials, where the two-body Hamiltonians have a virtual level at the bottom of the essential spectrum. In dimensions $d=2$ (in case of fermions) and $d=3$ the corresponding three-body Hamiltonian admits an infinite number of bound states, which is known as the Efimov effect. In this work we prove that this is not the case in higher dimensions. We investigate how the dimension and symmetries of the system influence this effect and prove the finiteness of the discrete spectrum of the corresponding three-body Hamiltonian.
\end{abstract}
\maketitle

\section{Introduction}
In the early seventies a large amount of literature was developed around the investigation of the discrete spectrum of many-body operators, such as \cite{Jafaev76} by D.R.~Yafaev  and \cite{Zhislin1974} by G.~Zhislin. Especially the three-body Hamiltonian has attracted a lot of interest in mathematics and physics since then. It is well known that if the two-body Hamiltonian has negative spectrum, then the three-body Hamiltonian has at most a finite number of negative eigenvalues \cite{Jafaev76}. Theorefore, from a mathematical point of view it was surprising what the physicists V.~Efimov has predicted \cite{Efimov}; he claimed that the three-body system in dimension three exhibits an infinite number of bound states, provided every two-body subsystem admits a positive spectrum and at least two of the three two-body subsystems possess a resonance at zero.

From a physical point of view for a long time it was unclear whether this so called Efimov effect could actually be observed experimentally. It became an outstanding challenge to observe this phenomenon. Thirty-five years after Efimov's prediction a group of experimental physicists observed these quantum states for the first time in an ultracold gas of caesium atoms \cite{Grimm}, where temparatures around $10 \,\mathrm{nK}$  were necessary to make the desired observation. Other observations of the Efimov states in different experiments followed later on, such as in the year 2009, where researchers reported experimental verification of the effect in bosonic quantum gases \cite{Phy2}.

The first rigorous mathematical proof of the Efimov effect was given by D.R.~Yafaev \cite{Jafaev} where he used a symmetrized form of Faddeev equations for the three-body Hamiltonian. He also verified that if at least two of the three two-body Hamiltonians do not have a resonance, then the Efimov effect does not occur \cite{Jafaev76}, which was also predicted in the original work of V.~Efimov. The first proof of this fact based on variational arguments is due to S.~Vugalter and G.~Zhislin \cite{Sem}. Later, this result was generalized in different directions, see \cite{Sem3},\cite{Sem4},\cite{Sem6},\cite{Sem7}.
The first variational proof of the Efimov effect is due to Yu. N. Ovchinnikov and I. M. Sigal \cite{Sigal}, which is based on the Born-Oppenheimer approximation. The proof has been improved later by H. Tamura \cite{Tamura2}. The technique of A. Sobolev \cite{Sobolev}, which uses similar methods as Yafaev's proof, combined with the low-energy behaviour of the resolvent \cite{Jensen3} and the calculation of the distribution of a Toeplitz operator is of great importance. Using this technique he established the low energy asymptotics 
\begin{equation*}
\lim\limits_{z\rightarrow 0^{-}} \frac{N(z)}{|\ln|z||}=\mathcal{U}_0 >0,
\end{equation*}
where $N(z)$ is the counting function of the eigenvalues of the three-body Hamiltonian below $z<0$. Later, H. Tamura improved this result by considering less restrictive pair potentials \cite{Tamura3}. The original work of V. Efimov also discussed whether the effect is possible in dimensions one and two or in three-dimensional subspaces with fixed symmetries. The absence of the Efimov effect in such symmetry subspaces was proved by S. Vugalter and G. Zhislin \cite{Zhislin}. They showed that the discrete spectrum of the Hamiltonian is always finite if one restricts its domain to functions underlying symmetries of irreducible representations of the group $S_3$ of permutation of particles. These subspaces are associated with nonzero angular momentum states or two or three identical fermions. The main reason of the finiteness of the discrete spectrum is due to the fact that in such subspaces a resonance state is always an eigenfunction corresponding to the eigenvalue zero. The existence and non-existence of the Efimov effect in lower dimensions was studied by the same authors in \cite{Sem5}. Under restrictive assumptions on the pair potentials and considered without symmetries they proved that the three-body system can admit at most a finite number of bound states in dimension two. In contrast to that, D.K.~Gridnev has recently proved the existence of the so called super-Efimov effect \cite{Gridnev} which at first was predicted in the physical work \cite{Nishida2}, where he considered a system of three spinless fermions in dimension two, each interacting through spherically-symmetric pair potentials. It turns out that such systems have two infinite series of bound states, each corresponding to the orbital angular momentum $l=\pm 1$ and
\begin{equation*}
\lim\limits_{z\rightarrow 0^{-}} \frac{N(z)}{|\ln|\ln z^2||}=\frac{8}{3\pi}.
\end{equation*}
We want to emphasize that in the case of two-dimensional fermions, the two-body resonance behaves like $c|x|^{-1}$ as $|x|\rightarrow \infty$, which is the borderline case of not being square-integrable. A similar situation occurs in dimension four, where the Hamiltonian is considered without symmetry restricitions. The corresponding resonance function decays as $c|x|^{-2}$, which is also on the edge of being in $L^2$. However, the physicists predicted that there is no Efimov effect in this case \cite{Nishida}. Our goal is to give rigorous mathematical proofs of these statements for any dimension greater than three. To the best of our knowledge it has not been studied in full detail so far. We highlight the differences between the behaviour of the corresponding counting functions and give a precise mathematical proof why the Efimov effect is absent in the four-dimensional case by using Sobolev's technique \cite{Sobolev} of a symmetrized form of the Faddeev equations. We study a more general concept of resonances called virtual levels and prove some of their fundamental characteristics, such as the impact on the behaviour of the resolvent of the two-body Hamiltonians and its consequences for the three-body system. This also leads to the finiteness of the discrete spectrum of the corresponding Hamiltonian in any dimension equal or greater than five.

 The paper is organized as follows.
 In Section~2 we introduce the notation for the three-body system and the corresponding two-body subsystems, and formulate the main results.
 In Section~3 we consider two-body Hamiltonians with a virtual level at the bottom of the essential spectrum. We describe the relation between virtual levels, resonances and actual bound states at zero energy, which depend not only on the dimension of the particles, but also on the corresonding symmetry.
 Section~4 is devoted to the four-dimensional case. We prove resolvent-related properties of resonances, which will be used in the proof of the main result.
 In Section~5 we prove the main results by combining the auxiliary statements from the previous sections.
\section{Notation and main result}
We consider a system of three particles with masses $m_1,m_2,m_3>0$ and corresponding position vectors $x_1,x_2,x_3\in \mathbb{R}^d, \ d \geq 4$.
The Hamiltonian of such a system in coordinate representation is given by
\begin{equation*}
-\frac{1}{2m_1}\Delta_{x_1}-\frac{1}{2m_2}\Delta_{x_2}-\frac{1}{2m_3}\Delta_{x_3}+v_{12}(x_{12})+v_{23}(x_{23})+v_{31}(x_{31}),
\end{equation*}
where $x_{ij}=x_i-x_j,\ i,j=1,2,3$, and $\Delta_{x_i}$ denotes the Laplacian with respect to coordinate $x_i=(x_{i1},\ldots,x_{id}) \in \mathbb{R}^d$ of the i-th particle. The real valued potential $v_\alpha, \ \alpha \in \{12,23,31\}$ describes the interaction of the corresponding particles with masses $m_1$ and $m_2$ or $m_2$ and $m_3$ or $m_3$ and $m_1$, respectively.
We assume that the potentials $v_\alpha$ satisfy
\begin{equation}\label{eq: assumptions on potential 1}
v_\alpha \in L_{\mathrm{loc}}^d(\mathbb{R}^d) \qquad \qquad \text{and} \qquad \qquad |v_\alpha (x)| \leq C|x|^{-b}, \ \text{if} \quad |x|\geq \gamma
\end{equation}
for some constant $\gamma>0$ and $b>2$. After separation of the center of mass the Hamiltonian of relative motion can be written as
\begin{equation}\label{eq: Hamiltonian in all dimensions}
H=H_0+\sum_{\alpha} v_{\alpha},
\end{equation}
where $H_0$ denotes the free Hamiltonian of the system. The corresponding configuration space $R_0$ is a $3(d-1)-$dimensional subspace of $\mathbb{R}^{3d}$. Under assumptions \eqref{eq: assumptions on potential 1} on the potentials $v_\alpha$ the operator $H$ is essentially self-adjoint. Every two-body subsystem corresponding to the subscript $\alpha\in\{12,23,31\}$ is described in the center of mass frame by the Hamiltonian
\begin{equation}
h_\alpha=-\frac{1}{2m_\alpha}\Delta+v_\alpha
\end{equation}
in $L^2(\mathbb{R}^d)$, where $m_\alpha$ is the reduced mass. Denote 
\begin{equation}
\mu = \min_\alpha \inf \sigma(h_\alpha),
\end{equation}
then by the HVZ-Theorem one has
\begin{equation}
\sigma_{\mathrm{ess}}(H)=[\mu,\infty).
\end{equation}
The case $\mu<0$ in dimension three was studied earlier \cite{Zhislin1974} and can be adapted to dimension $d \geq 4$. Hence, we only consider the case $\mu=0$. Our main results are the following two Theorems.
\begin{thm}\label{main result}
For $d=4$ let $v_\alpha(x) \leq 0$ as well as
\begin{equation*}
|v_\alpha(x)|\leq C(1+|x|)^{-b}, \ b>4,
\end{equation*}
and for $d\geq 5$ let $v_\alpha$ satisfy $\eqref{eq: assumptions on potential 1}$. Then $\sigma_{\mathrm{disc}}(H)$ is finite.
\end{thm}
\begin{rem}
In dimension $d=4$ we analyse the spectrum of the three-body Hamiltonian by the use of Faddeev equations following \cite{Sobolev}, which require more restrictions on the potentials. In case of $d\geq 5$ we make use of a variational type of argument following \cite{Zhislin}, which allows us to have more general assumptions on the potentials, namely \eqref{eq: assumptions on potential 1}.
In case of three identical particles the potentials $v_\alpha$ satisfy $v_\alpha(x_{ij})=v_\alpha(-x_{ij})$ and the operator $H$ is invariant under the action of the group $S_3$ of permutatation of particles. Denote by $\sigma_1,\sigma_2$ and $\sigma_3$ the three irreducible representations of $S_3$, where $\sigma_1$ is the symmetric representation, $\sigma_2$ the antisymmetric representation and $\sigma_3$ the two-dimensional irreducible representation, respectively. Denote by $P^{\sigma_i}, \ i=1,2,3$ the corresponding projection. Symmetries of types $\sigma_1$ and $\sigma_3$ do not put any restrictions on the symmetry of the two-particle subsystems. In case of $\sigma_2$ we denote the two-body Hamiltonians on $P^{\sigma_2}L^2$ by $h_\alpha^{\mathrm{as}}$ and the corresponding three-body Hamiltonian by $H^\mathrm{as}$. In dimensions $d\in \{2,3,4\}$ the operator $h_\alpha^{\mathrm{as}}$ has different properties of so called virtual levels (see later for a precise definition), which itself determines the existence or non-existence of the Efimov effect.
\end{rem}
\begin{thm}\label{main result 2}
For $d\geq 4$ let the potentials $v_\alpha$ satisfy \eqref{eq: assumptions on potential 1} and $v_\alpha(x)=v_\alpha(-x)$. Then $\sigma_{\mathrm{disc}}(H^{\mathrm{as}})$ is finite.
\end{thm}
\begin{rem}
Consequently, the super-Efimov effect does not exist in dimension $d\geq 4$.
\end{rem}
\section{Virtual levels of two-body subsystems}
The finiteness of the discrete spectrum of the three-body Hamiltonian depends on the existence and properties of resonances in the two-body subsystems, which appear in the frame of critical potentials and are sometimes called virtual levels of the Hamiltonian. We introduce the concept of virtual levels following \cite{Ja2}. For the sake of brevity we omit the subscript $\alpha$ in this section and add it for every corresponding expression only to distinguish between different subsystems, i.e. we consider the Schr\"odinger operator
\begin{equation*}
h= -\frac{1}{2m}\Delta +v.
\end{equation*}
This operator is given in the center of mass frame, acting in $L^2(\mathbb{R}^d), \ d\in \mathbb{N}$, where $m>0$ is the reduced mass and $v$ is the multiplication by a real valued function satisfying \eqref{eq: assumptions on potential 1}.
\begin{df}
The operator $h=-\frac{1}{2m}\Delta+v$ has a virtual level at the bottom of its spectrum, if
\begin{equation*}
h \geq 0 \qquad \qquad \text{and} \qquad \qquad \sigma_{\mathrm{disc}}\left(-\frac{1}{2m}\Delta +(1+\varepsilon)v\right) \not = \emptyset
\end{equation*}
holds for every $\varepsilon>0$.
\end{df}
The existence of a virtual level of $h$ is connected to the following homogeneous Sobolev-space
\begin{equation*}
\dot{H}^1(\mathbb{R}^d)= \overline{C_0^\infty(\mathbb{R}^d)}^{\Vert \cdot \Vert_{(1)}},\qquad \Vert f \Vert_{(1)}= \left(\int_{\mathbb{R}^d} |\nabla f(x)|^2 \, \mathrm{d}x \right)^{\frac{1}{2}}.
\end{equation*}
\begin{thm}\label{lem: virtual level implies rsonance}
Let $d\geq 4$. If the operator $h$ has a virtual level, then there exists a positive function $f \in \dot{H}^1(\mathbb{R}^d), f\not = 0$, such that 
\begin{equation}\label{eq: resonance or eigenfunction}
\left(-\frac{1}{2m}\Delta +v \right)f=0
\end{equation}
in the distributional sense.
\end{thm}
\begin{rem}
If $f\not \in L^2(\mathbb{R}^d)$, then $\lambda=0$ is called resonance of $h$ with the corresponding resonance state $f$.
\end{rem}
\begin{proof}
Let $k,n \in \mathbb{N}$, $B_k=\{x\in \mathbb{R}^d :\ |x| \leq k \}$ and
\begin{equation*}
h_0^k=-\frac{1}{2m} \Delta +v, \qquad \qquad h_n^k=-\frac{1}{2m}\Delta + \left(1+\frac{1}{n} \right)v,
\end{equation*}
considered as operators in $L^2(B_k)$ with form domains $H_0^1(B_k)$. Denote by $\lambda_0^k$ and $ \lambda_n^k$ the smallest Dirichlet-eigenvalue of the operator $h_0^k$ and $h_n^k$, respectively. Since
every domain $B_k$ is bounded we have $\lambda_0^k >0$ and $\lambda_n^k <0$ for every $n \in \mathbb{N}$ and sufficiently large $k\in \mathbb{N}$ due to the existence of the virtual level. For every sufficiently large $n \in \mathbb{N}$ we can find a $k_n \in \mathbb{N}$, such that $\lambda_n^{k_n} \geq 0$ and $\lambda_n^{k_n+1} <0$, which implies the existence of a region $A_n\subset \mathbb{R}^d$ and zero-eigenfunctions $u_n \in H_0^1(A_n)$. We normalize $u_n$ by
\begin{equation}\label{eq: normalization un}
\Vert u_n \Vert_{(1)}=\int_{A_n} |\nabla u_n(x)|^2 \, \mathrm{d}x = 1.
\end{equation}
Let
\begin{equation*}
f_n :\mathbb{R}^d \rightarrow \mathbb{R}, \ \ f_n(x)= 
\begin{cases} u_n(x), & x \in A_n
\\ 0, & x \in \mathbb{R}^d\backslash A_n
\end{cases}
\end{equation*}
Then we have $f_n \in H^1(\mathbb{R}^d)$ and $\Vert f_n\Vert_{(1)}=1$.
Due to \eqref{eq: normalization un} there exists a subsequence (also denoted by $f_n$) and a function $f$, such that $f_n \rightharpoonup f$ in $\dot{H}^1(\mathbb{R}^d)$, i.e.
\begin{equation*}
\int_{\mathbb{R}^d} \nabla f_n \cdot \overline{\nabla \varphi} \, \mathrm{d}x \overset{n \rightarrow \infty}{\longrightarrow} \int_{\mathbb{R}^d} \nabla f \cdot \overline{\nabla \varphi} \, \mathrm{d}x
\end{equation*}
holds for every $\varphi \in C_0^\infty(\mathbb{R}^d)$. Note that by assumptions \eqref{eq: assumptions on potential 1} the following operator is well defined.
\begin{equation*}
T_v: \dot{H}^1(\mathbb{R}^d) \rightarrow L^2(\mathbb{R}^d), \ \ \left( T_v u \right)(x)= |v(x)|^\frac{1}{2}u(x).
\end{equation*}
$T_v$ is compact, since for every $R>0$ the operator
\begin{equation*}
T_v^R : \dot{H}^1(\mathbb{R}^d) \rightarrow L^2(\mathbb{R}^d), \ \ \left( T_v^R u \right)(x)=|v(x)|^\frac{1}{2}u(x) \chi_{R}(x) 
\end{equation*}
is obviously compact and for sufficiently large $R>0$ we have
\begin{align*}
\Vert (T_v-T_v^R)u \Vert_{L^2(\mathbb{R}^d)}^2 &= \int_{\left\{|x|>R\right\}} |v(x)||u(x)|^2\, \mathrm{d}x = \int_{\left\{|x|>R\right\}} |v(x)||x|^2 \frac{|u(x)|^2}{|x|^2} \, \mathrm{d}x
\\ & \leq C\sup_{|y|>R} |y|^{2-b} \int_{\left\{|x|>R\right\}} |\nabla u(x)|^2\, \mathrm{d}x \overset{R\rightarrow \infty}{\longrightarrow} 0.
\end{align*}
Due to the compactness of $T_v$ the sequence $(T_v f_n)$ converges in $L^2(\mathbb{R}^d)$ and therefore
\begin{equation}\label{eq: f_n onverge to f}
\left(1+\frac{1}{n}\right) \int_{\mathbb{R}^d} v(x) f_n(x) \overline{\varphi(x)}\, \mathrm{d}x \longrightarrow \int_{\mathbb{R}^d} v(x) f(x) \overline{\varphi(x)}\, \mathrm{d}x
\end{equation}
holds for every $\varphi \in C_0^\infty(\mathbb{R}^d)$. Note that
\begin{align*}
\int_{\mathbb{R}^d} vf(f-\varphi)\, \mathrm{d}x-\int_{\mathbb{R}^d} v|f|^2\, \mathrm{d}x&=-\int_{\mathbb{R}^d} vf\varphi\, \mathrm{d}x
\\ &=\int_{\mathbb{R}^d} \nabla f\cdot \overline{\nabla \varphi}\, \mathrm{d}x= \int_{\mathbb{R}^d} |\nabla f |^2\, \mathrm{d}x - \int_{\mathbb{R}^d} \nabla f \cdot \left( \nabla f - \overline{\nabla \varphi}\right) \, \mathrm{d}x.
\end{align*}
Hence, by approximating $f$ with functions $\varphi_k \in C_0^\infty(\mathbb{R}^d)$, together with $\Vert |v|^\frac{1}{2} f\Vert \leq C \Vert |\nabla f|\Vert$ and $\Vert f \Vert_{(1)}=1$ we have $\int |v| |f|^2\, \mathrm{d}x =1$, which implies $f \not = 0$. Therefore $f \in \dot{H}^1(\mathbb{R}^d)$
is a positive solution of 
\begin{equation}\label{eq: a}
\left(-\frac{1}{2m}\Delta + v \right)f=0.
\end{equation}
This completes the proof.
\end{proof}
\begin{rem}
If the operator $h$ is considered without symmetry restrictions, then the solution $f$ is non-degenerate. Therefore, the proof of Theorem \ref{lem: virtual level implies rsonance} yields the following Corollary.
\end{rem}
\begin{cor}\label{Cor: proof}
There exists a constant $\mu>0$, such that for every function $g \in \dot{H}^1(\mathbb{R}^d)\backslash\{0\}$ with $\langle \nabla g,\nabla f\rangle=0$ one has
\begin{equation}\label{cor: estimate for g}
\langle (-\Delta+v)g,g\rangle \geq \mu \Vert |\nabla g| \Vert^2.
\end{equation}
\end{cor}
\begin{lem}\label{lem: vf in L^2}
Let $d\geq 4$ and let $v$ satisfy assumptions \eqref{eq: assumptions on potential 1}. Then for every $\psi\in \dot{H}^1(\mathbb{R}^d)$ one has $v\psi\in L^2(\mathbb{R}^d)$.
\end{lem}
\begin{proof}
We consider
\begin{equation*}
\int_{\mathbb{R}^d} |v(x)\psi(x)|^2\, \mathrm{d}x = \int_{\left\{|x|<\gamma\right\}} |v(x)\psi(x)|^2\, \mathrm{d}x + \int_{\left\{|x|\geq\gamma\right\}} |v(x)\psi(x)|^2\, \mathrm{d}x,
\end{equation*}
where $\gamma>0$ is the constant in \eqref{eq: assumptions on potential 1}. Since $v \in L_{\mathrm{loc}}^d(\mathbb{R}^d)$ and $\psi\in \dot{H}^1(\mathbb{R}^d)$ we obtain by the use of the Sobolev inequality
\begin{equation*}
\int_{\left\{|x|<\gamma\right\}} |v(x)\psi(x)|^2\, \mathrm{d}x \leq \left(\int_{\left\{|x|<\gamma\right\}} |v(x)|^d \, \mathrm{d}x\right)^{\frac{2}{d}} \left( \int_{\left\{|x|<\gamma\right\}} |\psi(x)|^{\frac{2d}{d-2}}\, \mathrm{d}x  \right)^{\frac{d-2}{d}}\leq c\Vert \psi \Vert_{(1)}^{\frac{2(d-2)}{d}}.
\end{equation*}
By the second assumption of \eqref{eq: assumptions on potential 1} we can make use of Hardy's inequality to conclude
\begin{equation*}
\int_{\left\{|x|\geq\gamma\right\}} |v(x)\psi(x)|^2\, \mathrm{d}x \leq c\int_{\left\{|x|\geq\gamma\right\}} \frac{|\psi(x)|^2}{|x|^2}\, \mathrm{d}x \leq C\Vert \psi \Vert_{(1)}^{2}.
\end{equation*}
Hence, $v\psi \in L^2(\mathbb{R}^d)$, which completes proof.
\end{proof}
\begin{lem}\label{lem: behaviour of the resonance state}
If $d=4$, then the solution $f$ of \eqref{eq: resonance or eigenfunction} belongs to the space 
\begin{equation*}
L_{-s}^2(\mathbb{R}^4)=\{ \varphi: \mathbb{R}^4 \rightarrow \mathbb{R} \ |\ (1+|\cdot|)^{-s}\varphi \in L^2(\mathbb{R}^4) \}
\end{equation*}
for every $s>0$. If $d\geq 5$, then $f \in L^2(\mathbb{R}^d)$.
\end{lem}
\begin{proof}
At first let $s_0\in(1,2)$. Then we have
\begin{equation}\label{eq: f in L_s_0}
\int_{\mathbb{R}^d} \frac{|f(x)|^2}{(1+|x|)^{2s_0}} \, \mathrm{d}x = \int_{\mathbb{R}^d} \frac{|x|^2}{(1+|x|)^{2s_0}} \frac{|f(x)|^2}{|x|^2} \, \mathrm{d}x \leq C \int_{\mathbb{R}^d} |\nabla f(x)|^2 \, \mathrm{d}x = C \Vert f \Vert_{(1)}^2,
\end{equation}
which implies $f\in L_{-s_0}^2(\mathbb{R}^d)$. By Lemma \ref{lem: vf in L^2} and \cite{Lieb} we have
\begin{equation}\label{eq: integralequation for resonance}
f(x) = G \ast \left( v\cdot f \right)(x) = -\frac{2m}{(d-2)w_d} \int_{\mathbb{R}^d} \frac{v(y)f(y)}{|x-y|^{d-2}} \, \mathrm{d}y,
\end{equation}
where $G$ is the fundamental solution of $\frac{1}{2m}\Delta$, which is given by
\begin{equation*}
G(x)= -\frac{2m}{(d-2)w_d}\frac{1}{|x|^{d-2}}
\end{equation*}
and $w_d=\frac{2\pi^{\frac{d}{2}}}{\Gamma\left(\frac{d}{2}\right)}$. We write
\begin{equation}\label{eq: decomposition f=f1+f2}
-\frac{(d-2)w_d}{2m} f=f_1+f_2,
\end{equation}
where
\begin{align*}
f_1(x)= \int\limits_{\left\{|x-y|< 2\right\}} \frac{v(y)f(y)}{|x-y|^{d-2}}\, \mathrm{d}y \qquad \text{and} \qquad f_2(x)= \int\limits_{\left\{|x-y| \geq 2\right\}} \frac{v(y)f(y)}{|x-y|^{d-2}}\, \mathrm{d}y.
\end{align*}
Since the function $x \mapsto |x|^{-(d-2)} \chi_{\{|x| <2\}}(x)$ belongs to $L^1(\mathbb{R}^d)$ and by Lemma \ref{lem: vf in L^2} we have $v f \in L^2(\mathbb{R}^d)$, Young's inequality implies $f_1 \in L^2(\mathbb{R}^d)$. Let us show that $f_2\in L_{-s}^2(\mathbb{R}^d)$ for every $s>0$. We rewrite assumption \eqref{eq: assumptions on potential 1} as
\begin{equation*}
|v(x)| \leq C|x|^{-(2+\theta)},
\end{equation*} 
where $|x|\geq\gamma$ and $\theta>0$. Note that we can always assume $\theta<1$. Let $1\leq q<2$, then by Lemma \ref{lem: vf in L^2} we have $vf\in L_{\mathrm{loc}}^q(\mathbb{R}^d)$. Hence, by the use of Hölder's inequality with $p_1=\frac{2}{q}$ and $p_2= \frac{2}{2-q}$ we obtain
\begin{align}\label{eq: virtual level last expression}
\Vert fv \Vert_{L^q(\mathbb{R}^d)}^q &= \int_{\mathbb{R}^d}|f(y)v(y)|^q \, \mathrm{d}y = \int_{\mathbb{R}^d} \frac{|f(y)|^q}{(1+|y|)^{qs_0}}|v(y)|^q(1+|y|)^{qs_0}\, \mathrm{d}y\notag
\\ &\leq C\left(\int_{\mathbb{R}^d} \frac{|f(y)|^2}{(1+|y|)^{2s_0}}\, \mathrm{d}y \right)^{\frac{1}{p_1}} \left( \int_{\mathbb{R}^d} \, (1+|y|)^{-p_2q(2+\theta-s_0)} \mathrm{d}y\right)^{\frac{1}{p_2}}.
\end{align}
By \eqref{eq: f in L_s_0} we have $f\in L_{-s_0}^2(\mathbb{R}^d)$. Therefore, \eqref{eq: virtual level last expression} is finite for
\begin{equation}\label{eq: assumption on q}
p_2q(2+\theta-s_0)>d \qquad \Leftrightarrow \qquad q> \frac{d}{2+\frac{d}{2}-(s_0-\theta)},
\end{equation}
which implies that $v f \in L^{q}(\mathbb{R}^d)$ for
\begin{equation*}
\frac{1}{q}< \frac{2}{d}+ \frac{\frac{d}{2}-(s_0-\theta)}{d}.
\end{equation*}
Since $x\mapsto  |x|^{-(d-2)}\chi_{\{|x|\geq 2\}}(x)$ belongs to $L^p(\mathbb{R}^d)$, where $p>\frac{d}{(d-2)}$, i.e. $\frac{1}{p}<1-\frac{2}{d}$, we conclude by Young's inequality that $f_2 \in L^r(\mathbb{R}^d)$, where
\begin{equation*}
r>\frac{d}{\frac{d}{2}-(s_0-\theta)}.
\end{equation*}
By the use of H\"older's inequality it follows $f_2 \in L_{-s_1}^2(\mathbb{R}^d)$, where $s_0-\theta<s_1<1$.
Subsequently applying this type of argument for $d=4$ yields $f \in L_{-s}^2(\mathbb{R}^4)$ for all $s>0$.  Let us show that in the case of $d\geq 5$ we have $f\in L^2(\mathbb{R}^d)$. Note that $2^{-1}<2^{-1}+2d^{-1}<1$ holds if and only if $d\geq 5$. Since $vf\in L^q(\mathbb{R}^d)$ for any $1 \leq q \leq 2$, we can choose 
\begin{equation}
2^{-1}+2d^{-1}<q^{-1}<1 \qquad \qquad \text{and} \qquad \qquad 2^{-1}<p^{-1}<1-2d^{-1},
\end{equation}
such that $q^{-1}+p^{-1}=1+2^{-1}$. Applying Young's inequality yields $f \in L^2(\mathbb{R}^d)$.
\end{proof}
\begin{rem}
In dimension $d=3$ (see \cite{Ja2}) the corresponding resonance state belongs to the weighted Sobolev space
\begin{equation*}
L_{-s}^2\left(\mathbb{R}^3\right)=\left\{\varphi:\mathbb{R}^3\rightarrow \mathbb{R}: (1+|\cdot|)^{-s}\varphi \in L^2(\mathbb{R}^3)  \right\},
\end{equation*}
where $s>\frac{1}{2}$. Our proof can be adapted to the case $d=3$, where in view of \eqref{eq: assumption on q} one has $s_0-\theta>\frac{1}{2}$, which leads to $f \in L_{-s}^2(\mathbb{R}^3), \ s > \frac{1}{2}$.
\end{rem}
\begin{lem}\label{lem: <v,f>=0}
If $d=4$, then the solution $f$ of \eqref{eq: resonance or eigenfunction} satisfies
\begin{equation}\label{eq: resonance behaviour}
f(x) = - \frac{m}{2\pi^2}\frac{\langle v,f\rangle}{|x|^{2}} + \widetilde{f}(x)
\end{equation}
as $|x| \rightarrow \infty$, where $\widetilde{f} \in L^2(\mathbb{R}^4)$.
\end{lem}
\begin{proof}
For $d=4$ we have $w_4=4\pi^2$ and in view of \eqref{eq: decomposition f=f1+f2}
\begin{equation*}
-\frac{2\pi^2}{m} f=f_1+f_2,
\end{equation*}
where $f_1 \in L^2(\mathbb{R}^4)$ and $f_2=g_2+h_2$, such that
\begin{equation*}
g_2(x)= \int\limits_{\left\{|x-y|\geq 2 \wedge |y|>|x|^\frac{1}{2}\right\}} \frac{v(y)f(y)}{|x-y|^2} \, \mathrm{d}y, \ \ \qquad h_2(x)=\int\limits_{\left\{|x-y|\geq 2 \wedge |y|\leq|x|^\frac{1}{2}\right\}}\frac{v(y)f(y)}{|x-y|^2} \, \mathrm{d}y.
\end{equation*}
Since we consider $|x|\rightarrow \infty$, we can set $|x| > \max\{4,\gamma^2\}$. By the use of
\begin{equation*}
|v(y)| \leq C(1+|y|)^{-2-\theta}
\end{equation*}
for $|y|>|x|^\frac{1}{2}\geq \gamma$ and $\theta>0$ we obtain
\begin{equation*}
|g_2(x)| \leq C\left(1+|x|^\frac{1}{2}\right)^{-\frac{\theta}{2}}\int\limits_{\left\{|x-y|\geq 2\right\}} \frac{|f(y)|}{(1+|y|)^{2+\frac{\theta}{2}}|x-y|^2}\, \mathrm{d}y.
\end{equation*}
Now since $f\in L_{-s}^2(\mathbb{R}^4)$ holds for all $s>0$, we have $(1+|\cdot|)^{-(2+\frac{\theta}{2})}f \in L^1(\mathbb{R}^4)$ and therefore $g_2 \in L^2(\mathbb{R}^4)$. For $|y|\leq|x|^{\frac{1}{2}}$ we have
\begin{equation*}
|x-y|^{-2}=|x|^{-2}\left|\frac{x}{|x|}-\frac{y}{|x|} \right|^{-2} \geq |x|^{-2}\left(1+|x|^{-\frac{1}{2}}\right)^{-2} = |x|^{-2}\frac{\left( 1-|x|^{-\frac{1}{2}}\right)^2}{\left( 1-|x|^{-1}\right)^2} \geq |x|^{-2}\left( 1-|x|^{-\frac{1}{2}}\right)^2.
\end{equation*}
On the other hand we can estimate 
\begin{equation*}
\ \ |x-y|^{-2}\leq \left(|x|-|y| \right)^{-2}=|x|^{-2} \left(\sum_{k=0}^\infty \left(\frac{|y|}{|x|}\right)^k\right)^2  \leq |x|^{-2}\left(1+2|x|^{-\frac{1}{2}} \right)^2.
\end{equation*}
This implies
\begin{equation*}
\frac{\langle v,f \rangle}{|x|^2}\left(1-|x|^{-\frac{1}{2}} \right)^2 \leq |h_2(x)|\leq \frac{\langle v,f \rangle}{|x|^2}\left(1+2|x|^{-\frac{1}{2}} \right)^2,
\end{equation*}
which completes the proof.
\end{proof}
\section{Resonance Interaction in dimension four}
In this section we consider $d=4$ and we further assume that the pair potentials $v_\alpha$ satisfy
\begin{equation}\label{eq: assumptions on potentials in dimension 4}
v_\alpha(x)\leq 0 \qquad \qquad \text{and} \qquad \qquad |v_\alpha(x)|\leq C(1+|x|)^{-b}, \ b>4.
\end{equation}
We will adapt the technique of \cite{Sobolev} to simplify the representation of $H$ and carry out the computations in the momentum space. By using this method we can also highlight the main reason why the Efimov effect is absent in this case. 

Following \cite{Fad} we denote by $k_i$ the conjugate variable of $x_i$ and introduce the set of variables $(k_\alpha,p_\alpha)$, conjugate with respect to the Jacobi-coordinates ($x_\alpha,y_\alpha)$, see \eqref{Jacobi in momentum 1}-\eqref{Jacobi in momentum 3} in the Appendix. The shift from $x_\alpha$ to $k_\alpha$ is carried out by the partial Fourier transform
\begin{equation}\label{eq:fouriertransform}
\left( \Phi_\alpha f\right)(k_\alpha, \cdot)= (2\pi)^{-2}\int_{\mathbb{R}^4} \mathrm{e}^{-\mathrm{i} \langle k_\alpha,x_\alpha \rangle} f(x_\alpha, \cdot) \, \mathrm{d}x_\alpha.
\end{equation}
Every pair of variables $k_\alpha,p_\alpha$ can be expressed by means of every other pair (see \eqref{Jacobi with p and k 1}-\eqref{Jacobi with p and k 3}). In this setting the three-particle Schr\"odinger operator (by abuse of notation) has the form
\begin{equation}\label{eq: hamiltonina in dimension four}
H=H_0+\sum_\alpha V_\alpha,
\end{equation}
where the kinetic energy is given by the multiplication of the function
\begin{equation*}
H_0f(k,p)= H^0(k,p)\cdot f(k,p),
\end{equation*}
where
\begin{equation}\label{eq: H^0 in all coordinates}
H^0(k,p)=\frac{k_\alpha^2}{2m_\alpha}+\frac{p_\alpha^2}{2n_\alpha} = \frac{k_\beta^2}{2m_\beta}+\frac{p_\beta^2}{2n_\beta} = \frac{k_\gamma^2}{2m_\gamma}+\frac{p_\gamma^2}{2n_\gamma}
\end{equation}
and the interactions are given by
\begin{equation*}
V_\alpha= \Phi_\alpha v_\alpha \Phi_\alpha^\ast.
\end{equation*}
Often it is useful to work with coordinates $(p_\alpha,p_\beta)$ instead of $(k_\alpha,p_\alpha)$. The relations are given by
\begin{equation}\label{eq: k_ab in p_a and p_b}
k_{\alpha}= d_{\alpha\beta}p_{\alpha}+e_{\alpha\beta}p_{\beta},
\end{equation}
where the constants $d_{\alpha\beta}$ and $d_{\alpha\beta}$ depend only on the masses $m_1,m_2$ and $m_3$ (see \eqref{Jacobi with p and k 1}-\eqref{Jacobi with p and k 3}). We denote by $H_{\alpha \beta}^0$ the function $H^0$ expressed in terms of $p_\alpha,p_\beta$, which then takes the form
\begin{equation*}
H_{\alpha \beta}^0(p_\alpha,p_\beta)=\frac{p_\alpha^2}{2m_\beta}+ \frac{\langle p_\alpha,p_\beta\rangle}{l_{\gamma}} + \frac{p_\beta^2}{2m_\alpha},
\end{equation*}
where $l_{\gamma} \in \{m_1,m_2,m_3\}$ (see \eqref{AP last}). By virtue of \eqref{eq: reduced masses1}-\eqref{eq: reduced masses2} it follows that
\begin{equation}\label{eq: H_0 etimate1}
H_{\alpha \beta}^0(p_\alpha,p_\beta) \geq \frac{p_\alpha^2}{2l_\alpha} + \frac{p_\beta^2}{2l_\beta}.
\end{equation}
Hence, by the use of the elementary Young's inequality one obtains
\begin{equation}\label{eq: H_0 etimate2}
H_{\alpha \beta}^0(p_\alpha,p_\beta) \geq c|p_{\alpha}|^{2\kappa}|p_{\beta}|^{2\kappa'}, \qquad \kappa+\kappa'=1.
\end{equation}
Following \cite{Jafaev} and \cite{Sobolev} we will use a symmetrized form of Faddeev equations to study the discrete spectrum of $H$. See in the Appendix for a detailed derivation of the Faddeev equations.
\begin{df}\label{def: definition of A(z)}
Let $z<0$ and
\begin{equation*}
A(z)=W^{\frac{1}{2}}(z)K(z)W^{\frac{1}{2}}(z),
\end{equation*}
where 
\begin{equation*}
W(z)=\begin{pmatrix}
W_{12}(z) & 0 & 0
\\ 0 & W_{23}(z) & 0
\\ 0 & 0 & W_{31}(z)
\end{pmatrix},
\end{equation*}
such that
\begin{equation}\label{eq: definition of w_alpha in three body case}
W_\alpha(z)= I + |V_\alpha|^\frac{1}{2}R_\alpha(z) |V_\alpha|^\frac{1}{2}, \qquad \qquad R_\alpha(z)=(H_0+V_\alpha-z)^{-1}
\end{equation}
and 
\begin{equation}\label{def: definition K(z)}
K(z)= \begin{pmatrix}
0 & K_{12|23}(z) & K_{12|31}(z)
\\ K_{23|12}(z) & 0 & K_{23|31}(z)
\\ K_{31|12}(z) & K_{32|23}(z) & 0
\end{pmatrix},
\end{equation}
such that
\begin{equation*}
K_{\alpha \beta}(z) = |V_\alpha|^{\frac{1}{2}}R_0(z) |V_\beta|^{\frac{1}{2}}, \qquad R_0(z)=(H_0-z)^{-1}.
\end{equation*}
\end{df}
The main property of $A(z)$ is the Birman-Schwinger-type characteristic. The proof of the following statement is given in \cite{Sobolev}.
\begin{thm}\cite[Theorem 4.1]{Sobolev}\label{thm counting functions of A and N}
Let $N(z)$ be the number of eigenvalues of the operator $H$ below $z<0$ and let $n(1,A(z))$ be the number of eigenvalues of the operator $A(z)$ greater than one. Then
\begin{equation*}
N(z)=n(1,A(z)).
\end{equation*}
\end{thm}
\begin{rem} In view of Theorem \ref{thm counting functions of A and N}, to prove Theorem \ref{main result} in the case $d=4$, it is sufficient to prove the compactness of $A(z)$ for $z \rightarrow 0$. In the presence of resonances in the two-body subsystems the corresponding operator $A_{\mathbb{R}^3}(z)$ in dimension three is not compact up to $z=0$, due to a singularity of $W_{\mathbb{R}^3}(z)$. To this end, we need to study the the operator $W_\alpha(z)$ in the two-body subsystems.
\end{rem}
\subsection*{Two-body subsystems in dimension four}
\begin{df}\label{df: w(z)}
Let $r_\alpha(z), z<0,$ be the resolvent of $h_\alpha=-\frac{1}{2m_\alpha}\Delta+v_\alpha$ and let
\begin{equation}\label{eq: definition of w}
w_\alpha(z)=I+|v_\alpha|^{\frac{1}{2}}r_\alpha(z)|v_\alpha|^{\frac{1}{2}}.
\end{equation}
\end{df}
Note that $w_\alpha(z)$ is uniformly bounded in $L^2(\mathbb{R}^d)$ for every $z\leq z_0<0$, where $|z_0|$ can be chosen arbitrarily small. Using the resolvent identity
\begin{equation*} 
r_{\alpha}(z)=r_0(z)-r_0(z)v_{\alpha} r(z)=r_0(z)-r_{\alpha}(z)v_{\alpha}r_0(z),
\end{equation*}
we have
\begin{align*}
I&=\left(I-|v_\alpha|^\frac{1}{2}r_0(z)|v_\alpha|^\frac{1}{2} \right)\left( I+|v_\alpha|^\frac{1}{2}r_\alpha(z)|v_\alpha|^\frac{1}{2}\right)
\\&=\left( I+|v_\alpha|^\frac{1}{2}r_\alpha(z)|v_\alpha|^\frac{1}{2}\right)\left(I-|v_\alpha|^\frac{1}{2}r_0(z)|v_\alpha|^\frac{1}{2} \right),
\end{align*}
which implies
\begin{equation}\label{eq: expression for w(z)}
w_\alpha(z)=I+|v_\alpha|^{\frac{1}{2}}r_\alpha(z)|v_\alpha|^{\frac{1}{2}}=\left( I-|v_\alpha|^\frac{1}{2}r_0(z)|v_\alpha|^\frac{1}{2}\right)^{-1}.
\end{equation}
Note that in accordance with Definition \ref{def: definition of A(z)} we have
\begin{equation}\label{eq: relation of W and w}
W_\alpha(z)=I+|V_\alpha|^\frac{1}{2}R_\alpha(z)|V_\alpha|^\frac{1}{2}=\Phi_\alpha w_\alpha \left( z- \frac{p_\alpha^2}{2n_\alpha} \right)\Phi_\alpha^\ast,
\end{equation}
where $\Phi_\alpha$ is the partial Fourier transform defined in \eqref{eq:fouriertransform}. The existence of a resonance of the two-body Hamiltonian $h_\alpha$ affects the behaviour of $w_\alpha(z)$ for $z\rightarrow 0$ (see \cite{Jensen4}). It produces a singularity of $w_\alpha(z)$ at $z=0$, which leads in dimension $d=3$ to the Efimov effect (see \cite{Sobolev}). We will see that in dimension four the singularity is not strong enough to break the compactness of $A(z)$ for $z\rightarrow 0$.
\begin{lem}\label{lem: resonance equiv to eigenfunction}
Let $G_\alpha$ be the integral operator with the kernel
\begin{equation}\label{eq kernel G_0}
G_\alpha(x,y)=\frac{m_\alpha}{2\pi^2} \frac{|v_\alpha(x)|^{\frac{1}{2}}|v_\alpha(y)|^\frac{1}{2}}{|x-y|^2},
\end{equation}
acting in $L^2(\mathbb{R}^4)$. If $\lambda=0$ is a resonance of $h_\alpha$, then $\mu=1$ is a simple eigenvalue of $G_\alpha$.
\end{lem}
\begin{proof}
By Lemma \ref{lem: <v,f>=0} the resonance is non-degenerate. Let $f$ be a resonance state of $h_\alpha$ and let $\varphi=|v_\alpha|^{\frac{1}{2}}f$. Then by Lemma \ref{lem: behaviour of the resonance state} we have $\varphi \in L^2(\mathbb{R}^4)$ and
\begin{align*}
\left(G_\alpha\varphi\right)(x)&= \frac{m_\alpha}{2\pi^2} \int_{\mathbb{R}^4} \frac{|v_\alpha(x)|^{\frac{1}{2}}|v_\alpha(y)|^{\frac{1}{2}}}{|x-y|^2}\varphi(y)\,\mathrm{d}y= |v_\alpha(x)|^{\frac{1}{2}}\left(-\frac{m_\alpha}{2\pi^2} \int_{\mathbb{R}^4} \frac{v_\alpha(y)f(y)}{|x-y|^2}\,\mathrm{d}y\right)
\\ &= |v_\alpha(x)|^{\frac{1}{2}}f(x)= \varphi(x).
\end{align*}
\end{proof}
\begin{lem}\label{lem: free resolvent}
Let $G_\alpha$ be the operator defined by the kernel \eqref{eq kernel G_0}. For $z<0$, $|z|$ sufficiently small, there exist compact operators $G_1$, $G_2$ and a constant $\delta>0$, such that
\begin{equation*}
|v_\alpha|^\frac{1}{2}r_0(z)|v_\alpha|^\frac{1}{2}=G_\alpha+zG_1+z\ln|z|G_2+|z|^{1+\delta}G_\alpha^{(\delta)}(z),
\end{equation*}
where $G_\alpha^{(\delta)}(z)$ is an operator, such that $\Vert G_\alpha^{(\delta)}(z) \Vert_{HS} \leq C_\delta |z|^\delta $, where $\Vert \cdot \Vert_{HS}$ denotes the Hilbert-Schmidt norm.
\end{lem} 
\begin{proof}
In the following we consider $|z|<1$. The kernel of $(-\Delta-z)^{-1}$ is given by
\begin{equation*}
(-\Delta-z)^{-1}(|x-y|)=\frac{\mathrm{i}\sqrt{z}}{8\pi |x-y|}H_1^{(1)}\left( \sqrt{z} |x-y|\right), \ x,y,\in \mathbb{R}^4,
\end{equation*}
where $H_1^{(1)}$ is the first Hankel function (see \cite{Abramowitz}). Hence,
\begin{align}
r_0(z,|x-y|)&=\left( -\frac{1}{2m_\alpha}\Delta -z \right)^{-1}\left(|x-y|\right)=2m_\alpha\left( -\Delta -2m_\alpha z \right)^{-1}(|x-y|) \notag
\\ &= \frac{m_\alpha\mathrm{i}\sqrt{2 m_\alpha z}}{4 \pi |x-y|} H_1^{(1)} \left( \sqrt{2m_\alpha z} |x-y|\right).\label{eq: free resolvent with hankel}
\end{align}
According to \cite{Abramowitz}, p.$360$, one has $H_1^{(1)}(\zeta)=J_1(\zeta)+\mathrm{i}Y_1(\zeta)$ and
\begin{align*}
J_1(\zeta)&=\frac{\zeta}{2}\sum_{k=0}^\infty \frac{\left( -\frac{1}{4}\zeta^2 \right)^k}{k!(k+1)!},
\\ Y_1(\zeta)&= -\frac{2}{\pi \zeta} +\frac{2}{\pi}\ln\left( \frac{\zeta}{2} \right)J_1(\zeta)-\frac{\zeta}{2\pi} \sum_{k=0}^\infty \left( \psi(k+1)+\psi(k+2) \right)\frac{\left( -\frac{1}{4}\zeta^2 \right)^k}{k! \left(k+1 \right)!},
\end{align*}
where
\begin{equation*}
\psi(1)=-1 , \qquad \psi(k)= \sum_{j=1}^{k-1} \frac{1}{j}- \gamma, \ k \geq 2
\end{equation*}
and $\gamma$ is the Euler–Mascheroni constant. Hence, we obtain
\begin{equation}\label{eq: series expansion of hankel}
H_1^{(1)}(\zeta)= -\frac{2\mathrm{i}}{\pi \zeta}+ \left( \frac{\zeta}{2}+\frac{\zeta \mathrm{i}}{\pi}\ln\left( \frac{\zeta}{2}\right) \right) \sum_{k=0}^\infty a_k\left( \zeta^2\right)^k- \frac{\zeta\mathrm{i}}{2\pi} \sum_{k=0}^\infty b_k \left( \zeta^2\right)^k,
\end{equation}
where
\begin{equation*}
a_k=\frac{(-1)^k}{4^k k!(k+1)!} \qquad \text{and} \qquad b_k=\left(
\psi(k+1)+\psi(k+2)\right)a_k.
\end{equation*}
Note that both series in \eqref{eq: series expansion of hankel} converge for every $\zeta \in \mathbb{C}$. By \eqref{eq: free resolvent with hankel} and \eqref{eq: series expansion of hankel} we get
\begin{align}
|v_\alpha|^\frac{1}{2}(x)r_0(z,|x-y|)|v_\alpha|^\frac{1}{2}(y)&=|v_\alpha|^\frac{1}{2}(x)|v_\alpha|^\frac{1}{2}(y)\frac{m_\alpha\mathrm{i}\sqrt{2 m_\alpha z}}{4 \pi |x-y|} H_1^{(1)} \left( \sqrt{2m_\alpha z} |x-y|\right) \notag
\\ &= G_\alpha(x,y)+zG_1(x,y)+z\ln|z|G_2(x,y)+ G(x,y,z),\label{eq: free resolvent with remainder}
\end{align}
where
\begin{align}\label{eq: G0}
G_\alpha(x,y)&= \frac{m_\alpha}{2\pi^2}\frac{|v_\alpha|^\frac{1}{2}(x)|v_\alpha|^\frac{1}{2}(y)}{|x-y|^2},
\\ G_1(x,y)&= \frac{m_\alpha^2}{4\pi^2}|v_\alpha|^\frac{1}{2}(x)|v_\alpha|^\frac{1}{2}(y)\left(\psi(1)+\psi(2)-\ln(2m_\alpha)-2\ln\left( \frac{|x-y|}{2} \right) \right), \label{eq: G1}
\\ G_2(x,y)&= -\frac{m_\alpha^2}{4\pi^2}|v_\alpha|^\frac{1}{2}(x)|v_\alpha|^\frac{1}{2}(y). \label{eq: G2}
\end{align}
We will show that the remainder $G(x,y,z)$ is a Hilbert-Schmidt kernel and that the Hilbert-Schmidt norm is of order $\mathcal{O}\left(|z|^{1+\delta}\right)$ as $z \rightarrow 0$, where $\delta>0$ is sufficiently small. 
\\Let $\sqrt{2m_\alpha|z|}|x-y|> 1$. By \cite{Abramowitz}, p.364, we have
\begin{equation}\label{eq: estimate abram}
\left| H_1^{(1)}(\zeta)\right| \leq c |\zeta|^{-\frac{1}{2}} , \ |\zeta| \geq 1.
\end{equation}
Relations \eqref{eq: free resolvent with hankel} and \eqref{eq: free resolvent with remainder} imply
\begin{align*}
&|G(x,y,z)|\chi_{\{ \sqrt{2m_\alpha|z|}|x-y|> 1 \}} \notag
\\& \ \ \ \ \ \ \ \ \ \leq c \frac{|z|^\frac{1}{2}|v_\alpha|^\frac{1}{2}(x)|v_\alpha|^\frac{1}{2}(y)}{|x-y|}\left|H_1^{(1)}(\sqrt{2m_\alpha z}|x-y|) \right|+|G_\alpha|+|z||G_1|+|z\ln|z|||G_2|.
\end{align*}
Hence, by definition of the kernels \eqref{eq: G0}-\eqref{eq: G2} together with \eqref{eq: estimate abram} we have
\begin{align}
&|G(x,y,z)|\chi_{\{ \sqrt{2m_\alpha |z|}|x-y|> 1 \}} \notag
\\ & \ \ \ \ \ \ \ \ \  \leq |z||\ln|z|||v_\alpha|^\frac{1}{2}(x)|v_\alpha|^\frac{1}{2}(y)\left(c_1+c_2\left|\ln\left(\frac{|x-y|}{2} \right)\right|\right) \notag
\\ & \ \ \ \ \ \ \ \ \  \leq |z||\ln|z|||v_\alpha|^\frac{1}{2}(x)|v_\alpha|^\frac{1}{2}(y)\frac{(1+|x|)^{4\delta}(1+|y|)^{4\delta}}{(1+|x-y|)^{4\delta}}\left(c_1+c_2\left|\ln\left(\frac{|x-y|}{2} \right)\right|\right) \notag
\\ & \ \ \ \ \ \ \ \ \  \leq c|z|^{1+2\delta}|\ln|z||v_\alpha|^\frac{1}{2}(x)|v_\alpha|^\frac{1}{2}(y)|(1+|x|)^{4\delta}(1+|y|)^{4\delta}\left(c_1+c_2\left|\ln\left(\frac{|x-y|}{2} \right)\right|\right).\label{eq: estimate for G for >1}
\end{align}
Now let $\sqrt{2m_\alpha |z|}|x-y|\leq 1$. Note that in view of \eqref{eq: series expansion of hankel} we have
\begin{equation}\label{eq: remainder}
G(x,y,z)=z|v_\alpha|^\frac{1}{2}(x)\left(\sum_{j=1}^\infty\sum_{k=0}^1 z^j (\ln|z|)^k G_j^k(x,y)\right)|v_\alpha|^\frac{1}{2}(y),
\end{equation}
where the kernels $G_j^k$ are defined by
\begin{equation}
G_j^1(x,y)= -(2m_\alpha)^j\alpha_j|x-y|^{2j},\qquad G_j^0(x,y) = (2m_\alpha)^j\alpha_j|x-y|^{2j} \left(\beta_j-2\ln\left( \frac{|x-y|}{2} \right)\right)
\end{equation}
and the constants $\alpha_j,\beta_j$ are given by
\begin{equation*}
\alpha_j=\frac{(-1)^jm_\alpha^2}{4^{j+1}\pi^2j!(j+1)!}, \qquad \beta_j=\psi(j+1)+\psi(j+2)-\ln(2m_\alpha).
\end{equation*}
By definition of the kernels $G_j^k$ we have
\begin{equation}\label{eq: G(z) with sigma}
G(x,y,z)= z|v_\alpha|^\frac{1}{2}(x)\left(\sigma_1(x,y,z)+\sigma_2(x,y,z) +\sigma_3(x,y,z)\right)|v_\alpha|^\frac{1}{2}(y),
\end{equation}
where
\begin{align*}
\sigma_1(x,y,z)&= \sum_{j=1}^\infty \alpha_j \beta_j \left( \sqrt{2m_\alpha z}|x-y| \right)^{2j}, \\\sigma_2(x,y,z)&= -2\ln\left(\frac{|x-y|}{2} \right)\sum_{j=1}^\infty \alpha_j \left( \sqrt{2m_\alpha z}|x-y| \right)^{2j},
\\ \sigma_3(x,y,z)&= -\ln|z|\sum_{j=1}^\infty \alpha_j \left( \sqrt{2m_\alpha z}|x-y| \right)^{2j}.
\end{align*}
We are going to estimate $\sigma_1, \sigma_2$ and $\sigma_3$  separately. Let $0<\delta<2^{-1}$. Since $\sqrt{2m_\alpha |z|}|x-y|\leq 1$, we have
\begin{align}\label{eq: sigma1}
|\sigma_1(x,y,z)| &\leq \left( \sqrt{2m_\alpha |z|}|x-y|\right)^{4\delta} \sum_{j=1}^\infty |\alpha_j\beta_j| \left( \sqrt{2m_\alpha |z|}|x-y| \right)^{2(j-2\delta)} \notag
\\ &\leq C|z|^{2\delta}|x-y|^{4\delta} \sum_{j=1}^\infty |\alpha_j \beta_j| \leq C_1 |z|^{2\delta}(1+|x|)^{4\delta}(1+|y|)^{4\delta}.
\end{align}
In the last inequality we used the fact that $\sum\limits_{j=1}^\infty |\alpha_j \beta_j|<\infty$.
Analogously we obtain
\begin{align}\label{eq: sigma2}
|\sigma_2(x,y,z)| &\leq 2\left|\ln\left(\frac{|x-y|}{2} \right)\right|\sum_{j=1}^\infty |\alpha_j| \left( \sqrt{2m_\alpha|z|}|x-y| \right)^{2j}  \notag
\\ &\leq C_2 |z|^{2\delta}\left|\ln\left(\frac{|x-y|}{2} \right)\right|(1+|x|)^{4\delta}(1+|y|)^{4\delta}
\end{align}
and also
\begin{equation}\label{eq: sigma3}
|\sigma_3(x,y,z)| \leq C_3|z|^{2\delta}|\ln|z||(1+|x|)^{4\delta}(1+|y|)^{4\delta}. \ \ \ \ \ \ \ \ \ \ \ \ \ \ \ \ \  \ \ \ \
\end{equation}
Hence, by collecting estimates \eqref{eq: sigma1}-\eqref{eq: sigma3} together with \eqref{eq: G(z) with sigma} we get for $|z|<1$ sufficiently small
\begin{align}\label{eq: estimate for G for <1}
&|G(x,y,z)|\chi_{\{ \sqrt{2m_\alpha|z|}|x-y|\leq 1 \}}  \notag
\\ & \ \ \ \ \ \ \ \ \ \ \ \ \ \ \ \ \ \leq |z|^{1+2\delta}|\ln|z|||v_\alpha|^\frac{1}{2}(x)|v_\alpha|^\frac{1}{2}(y)(1+|x|)^{4\delta}(1+|y|)^{4\delta}\left(c_3+c_4\left|\ln\left(\frac{|x-y|}{2} \right)\right|\right).
\end{align}
By combining estimates \eqref{eq: estimate for G for <1} and \eqref{eq: estimate for G for >1} we get
\begin{equation*}
|G(x,y,z)| \leq C|z|^{1+2\delta}|\ln|z|| |v_\alpha|^\frac{1}{2}(x)|v_\alpha|^\frac{1}{2}(y)(1+|x|)^{4\delta}(1+|y|)^{4\delta}\left(1+\left|\ln\left(\frac{|x-y|}{2} \right)\right|\right).
\end{equation*}
Since
\begin{equation*}
\left|\ln\left(\frac{|x-y|}{2}\right)\right| \leq C \max\left\{ |x-y|^\varepsilon, |x-y|^{-\varepsilon} \right\}, \ \varepsilon>0
\end{equation*}
and 
\begin{equation*}
|v_\alpha(x)|\leq C(1+|x|)^{-b}, \ b>4,
\end{equation*}
we can choose $\varepsilon,\delta>0$, such that $0<\delta<\frac{b-4-2\varepsilon}{8}$, which implies that the remainder $G(z)$ is Hilbert-Schmidt and the operator norm is of order $\mathcal{O}\left( |z|^{1+2\delta}|\ln|z|| \right)$. Hence, the operator $G_\alpha^{(\delta)}(z)= |z|^{-1-\delta}G(z)$ is bounded up to $z\leq 0$. Further, we have
\begin{equation*}
|v_\alpha|^\frac{1}{2}r_0(z)|v_\alpha|^\frac{1}{2}=G_\alpha+zG_1+z\ln|z|G_2+|z|^{1+\delta} G_\alpha^{(\delta)}(z), \ \delta >0,
\end{equation*}
which completes the proof.
\end{proof}
\begin{rem}
We used similar arguments as in \cite{Jensen4}, where it was shown that for
\begin{equation*}
|v_\alpha(x)|\leq C(1+|x|)^{-b}, \ b>8,
\end{equation*}
$G(z)$ is of order $\mathcal{O}(|z|^2\ln|z|)$. We allow weaker assumptions on the potential and obtain a weaker estimate as a result. By using arguments from \cite{Jensen4} we can prove the following Lemma.
\end{rem}
\begin{lem}\label{lem: expansion of w}
If $\lambda=0$ is a resonance of $h_\alpha$, then for $z<0$, $|z|$ sufficiently small, the operator $w_\alpha(z)$ has the representation
\begin{equation}\label{eq: expansion of w}
w_\alpha(z)= (z(\ln|z|-\tau_\alpha))^{-1}\langle \cdot , \varphi \rangle \varphi + \left(z(\ln|z|-\tau_\alpha)\right)^{-1+\delta} w_\alpha^{(\delta)}(z),
\end{equation}
where $\delta>0$ is sufficiently small, $\varphi$ is an eigenfunction of the operator $G_\alpha$ corresponding to the eigenvalue $\mu=1$ and $\tau_\alpha \in \mathbb{R}$ is a constant, which depends on the potential $v_\alpha$ and the mass $m_\alpha$. In addition, the operator $w_\alpha^{(\delta)}(z)$ is bounded for $z\leq0$.
\end{lem}
\begin{proof}
Let
\begin{equation*}
s_\alpha(z)=I-|v_\alpha|^\frac{1}{2}r_0(z)|v_\alpha|^{\frac{1}{2}}.
\end{equation*}
We will use the expansion of Lemma \ref{lem: free resolvent} in order to compute the inverse
\begin{equation}
s_\alpha^{-1}(z)=w_\alpha(z).
\end{equation}
Let $P_0$ be the one-dimensional projection on the subspace associated with the eigenfunction $\varphi$ of the operator $G_\alpha$ corresponding to the eigenvalue $\mu=1$ (c.f. Lemma \ref{lem: resonance equiv to eigenfunction}) and denote by $P_1$ the projection onto the orthogonal complement of the eigenspace of $\mu$ in $L^2(\mathbb{R}^4)$. Following \cite{Jensen4}, for every $\psi \in L^2(\mathbb{R}^4)$ we have the unique decomposition $\psi=P_0\psi+P_1\psi$, which allows us to write $s_\alpha(z)\psi$ as $S(z)(P_0\psi,P_1\psi)$, where
\begin{equation}
S(z)=
\begin{pmatrix}
P_0s_\alpha(z)P_0 & P_0s_\alpha(z)P_1
\\ P_1s_\alpha(z)P_0 & P_1s_\alpha(z)P_1
\end{pmatrix}.
\end{equation}
Further, let
\begin{equation*}
P(z)=\begin{pmatrix}
|z|^{-\frac{1}{2}} P_0 & 0
\\ 0 & P_1
\end{pmatrix} \qquad \text{and} \qquad B(z)=P(z)S(z)P(z).
\end{equation*}
The entries of $B(z)$ are given by
\begin{align*}
b_{00}(z)&= |z|^{-1}P_0(I-|v_\alpha|^{\frac{1}{2}}r_0(z)|v_\alpha|^{\frac{1}{2}})P_0, \qquad \,b_{01}(z)= |z|^{-\frac{1}{2}}P_0(I-|v_\alpha|^{\frac{1}{2}}r_0(z)|v|^{\frac{1}{2}})P_1,
\\ b_{10}(z)&= |z|^{-\frac{1}{2}}P_1(I-|v_\alpha|^{\frac{1}{2}}r_0(z)|v_\alpha|^{\frac{1}{2}})P_0, \qquad b_{11}(z)= P_1(I-|v_\alpha|^{\frac{1}{2}}r_0(z)|v_\alpha|^{\frac{1}{2}})P_1.
\end{align*}
By the use of Lemma \ref{lem: free resolvent} we have $B(z)=C(z)+D(z)$, where
\begin{align*}
C(z)&=\begin{pmatrix}
P_0\left(G_1+\ln|z|G_2\right)P_0 & 0
\\ 0 & P_1\left( I-G_\alpha \right)P_1
\end{pmatrix} \quad \text{and} \quad D(z)=
\begin{pmatrix}
d_{00}(z) & d_{01}(z) 
\\ d_{10}(z)  & d_{11}(z) 
\end{pmatrix}
\end{align*}
The entries of $D(z)$ are given by
\begin{align*}
d_{00}(z)  &= -|z|^\delta P_0G_\alpha^{(\delta)}(z)P_0,\\
 d_{01}(z) &= |z|^\frac{1}{2}P_0\left( G_1+\ln|z|G_2-|z|^\delta G_\alpha^{(\delta)}(z) \right)P_1,
\\ d_{10}(z)  &= |z|^\frac{1}{2}P_1\left( G_1+\ln|z|G_2-|z|^\delta G_\alpha^{(\delta)}(z) \right)P_0,\\ d_{11}(z)  &= |z| P_1\left(G_1+\ln|z|G_2-|z|^\delta G_\alpha^{(\delta)}(z) \right)P_0.
\end{align*}
By abuse of notation we write
\begin{equation}\label{eq. abuse of notation for D}
D(z)=\mathcal{O}\left( |z|^\delta \right).
\end{equation}
Since $P_1$ projects onto the subspace of functions orthogonal to $\varphi$, the operator $P_1(I-G_\alpha)P_1$ is invertible. Note that by Lemma \ref{lem: resonance equiv to eigenfunction} we have $\langle |v|^\frac{1}{2},\varphi \rangle \not = 0$. Hence, we can normalize $\varphi$, such that
\begin{equation}\label{eq: normalization phi}
\langle |v|^\frac{1}{2},\varphi \rangle = \frac{2\pi}{m_\alpha}.
\end{equation}
Then we have
\begin{equation*}
\langle G_2\varphi, \varphi\rangle = -1 \qquad \qquad \text{and} \qquad \qquad \langle G_1\varphi,\varphi \rangle=\tau_\alpha,
\end{equation*}
where due to \eqref{eq: G1}
\begin{equation*}
\tau_\alpha= \frac{m_\alpha^2}{4\pi^2}\int\int \left(\psi(1)+\psi(2)-\ln(2m_\alpha)-2\ln\left(\frac{|x-y|}{2} \right)\right)|v_\alpha(x)|^{\frac{1}{2}}|v_\alpha(y)|^{\frac{1}{2}}\varphi(x) \varphi(y)\, \mathrm{d}x\mathrm{d}y.
\end{equation*}
Using $P_0=\Vert \varphi \Vert^{-2}\langle \cdot , \varphi\rangle \varphi$ we obtain
\begin{align*}
P_0\left(G_1+\ln|z|G_2\right)P_0  = \frac{(\tau_\alpha - \ln|z|)}{\Vert \varphi \Vert^2}P_0
\end{align*}
and therefore
\begin{equation*}
C^{-1}(z)=\begin{pmatrix}
\frac{\langle \cdot,\varphi\rangle \varphi}{(\tau_\alpha -\ln|z|)} & 0
\\ 0 & K
\end{pmatrix},
\end{equation*}
where $K=\left(P_1(I-G_\alpha)P_1\right)^{-1}$. Now we can write
\begin{equation*}
B(z)=C(z)+D(z)=\left(I+D(z) C^{-1}(z)\right)C(z).
\end{equation*}
By \eqref{eq. abuse of notation for D} we have
\begin{equation*}
\Vert D(z)C^{-1}(z) \Vert \overset{z\rightarrow 0}{\longrightarrow} 0.
\end{equation*}
Therefore, we obtain the inverse of $B(z)$ by the Neumann series
\begin{equation*}
B^{-1}(z)=C^{-1}(z)\left( I-\left( -D(z)C^{-1}(z) \right) \right)^{-1}=C^{-1}(z)+C^{-1}(z)\sum_{n=1}^\infty \left(-D(z)C^{-1}(z)\right)^n.
\end{equation*}
Note that
\begin{equation*}
\sum_{n=1}^\infty \Vert D(z)C^{-1}(z)\Vert^n \leq \frac{\Vert D(z)C^{-1}(z)\Vert}{1-\Vert D(z)C^{-1}(z)\Vert},
\end{equation*}
which together with \eqref{eq. abuse of notation for D} yields
\begin{equation*}
B^{-1}(z) = \begin{pmatrix}
\frac{\langle \cdot,\varphi\rangle \varphi}{(\tau_\alpha -\ln|z|)} & 0
\\ 0 & K
\end{pmatrix} + \mathcal{O}\left(|z|^{\delta}\right).
\end{equation*}
Note that
\begin{equation*}
S^{-1}(z)=P(z)B^{-1}(z)P(z)
\end{equation*}
and $|z|(\tau_\alpha - \ln|z|) = z(\ln|z|-\tau_\alpha)$ for sufficiently small $|z|$. This completes the proof.
\end{proof}
The proof of the next Lemma follows from similar arguments as in \cite{Sobolev}. We adapt the proof to our case.
\begin{lem}\label{lem: sqrt w(z)}
For $z<0$, $|z|$ sufficiently small, the operator $w_\alpha(z)$ is positive and we have
\begin{equation*}
w_\alpha^{\frac{1}{2}}(z)= \frac{\langle \cdot , \varphi \rangle \varphi}{\Vert \varphi \Vert \sqrt{z\left(\ln|z|-\tau_\alpha\right)}} + (z(\ln|z|-\tau_\alpha))^{-\frac{1-\delta}{2}}\tilde{w}_\alpha^{(\delta)}(z),
\end{equation*}
where $\tilde{w_\alpha}^{(\delta)}(z)$ is bounded for $z\leq 0$.
\begin{proof}
For $|z|$ sufficiently small one has $z(\ln|z|-\tau_\alpha)>0$ Hence, by $P_0=P_0^2$ we have
\begin{equation*}
\left( \frac{\langle \cdot , \varphi \rangle \varphi}{z(\ln|z|-\tau_\alpha)} \right)^{\frac{1}{2}} = \left(\frac{\Vert \varphi \Vert^2P_0^2}{z(\ln|z|-\tau_\alpha)}\right)^{\frac{1}{2}}=\frac{\langle \cdot , \varphi \rangle \varphi}{\Vert \varphi \Vert\sqrt{z(\ln|z|-\tau_\alpha)}}.
\end{equation*}
Since $r_\alpha(z)\geq 0$ for $h_\alpha\geq0$ we have $w_\alpha(z) \geq I \geq 0$. By the use of 
\begin{equation*}
\Vert A^{\frac{1}{2}}-B^{\frac{1}{2}} \Vert \leq \Vert A - B \Vert^{\frac{1}{2}}
\end{equation*}
for positive operators $A,B$ we obtain from Lemma \ref{lem: expansion of w}
\begin{equation*}
\left\Vert w_\alpha^{\frac{1}{2}}(z)-\frac{\langle \cdot , \varphi \rangle \varphi}{\Vert \varphi \Vert \sqrt{z\left(\ln|z|-\tau_\alpha\right)}} \right\Vert \leq C (z(\ln|z|-\tau_\alpha))^{-\frac{1-\delta}{2}},
\end{equation*}
which completes the proof.
\end{proof}
\end{lem}
\subsection*{Three-body system in dimension four}
Now we move to the three-body system. In this section we will prove that every entry $A_{\alpha\beta}(z)$ of the matrix $A(z)$ is a compact operator for every $z\leq0$. By Definition \ref{def: definition of A(z)} we have
\begin{equation*}
A_{\alpha\beta}(z) = W_\alpha^{\frac{1}{2}}(z) K_{\alpha \beta}(z) W_\beta^{\frac{1}{2}}(z).
\end{equation*}
Due to the partial Fourier transform $\Phi_\alpha,\Phi_\beta$, defined by \eqref{eq:fouriertransform}, and the structure of the operator $A_{\alpha\beta}(z)$, we will make use of the mixed coordinates $(x_\alpha,p_\alpha), (x_\beta,p_\beta)$ and various relations such as \eqref{eq: H^0 in all coordinates} and \eqref{eq: k_ab in p_a and p_b}. We start with the proof of the compactness of $K_{\alpha\beta}(z)$ by adapting the proof of \cite{Sobolev} to our case.
\begin{lem}\label{lem: K compact}
The operator $K_{\alpha\beta}(z)$ is compact for every $z\leq 0$.
\end{lem}
\begin{proof}
In view of \eqref{def: definition K(z)} the operator $K_{\alpha\beta}(z)$ is given by
\begin{equation*}
K_{\alpha \beta}(z) = |V_\alpha|^{\frac{1}{2}}R_0(z) |V_\beta|^{\frac{1}{2}}.
\end{equation*}
It is sufficient to consider the operator
\begin{equation*}
\tilde{K}_{\alpha\beta}(z)=\Phi_\alpha^\ast K_{\alpha \beta}(z)\Phi_\beta.
\end{equation*}
For $R \geq 1$ let
\begin{equation*}
\chi_R: \mathbb{R}^4 \rightarrow \mathbb{R},\chi_R(p)=\begin{cases}
1, & |p|\leq R
\\ 0, &|p|>R
\end{cases}
\end{equation*}
We decompose $\tilde{K}_{\alpha\beta}(z)$ as
\begin{equation*}
\tilde{K}_{\alpha\beta}(z) = Z_{\alpha\beta}^R(z) + Y_{\alpha\beta}^R(z),
\end{equation*}
where
\begin{align*}
Z_{\alpha\beta}^R(z)&=\chi_R(p_\alpha) \tilde{K}_{\alpha\beta}\chi_R(p_\beta),
\\
Y_{\alpha\beta}^R(z)&=  (I-\chi_R(p_\alpha)) \tilde{K}_{\alpha\beta}\chi_R(p_\beta) +  \chi_R(p_\alpha) \tilde{K}_{\alpha\beta}(I-\chi_R(p_\beta)) + (I-\chi_R(p_\alpha)) \tilde{K}_{\alpha\beta}(I-\chi_R(p_\beta)).
\end{align*}
The kernel of the operator $Z_{\alpha\beta}^R(z)$ is square-integrable for $z \leq 0$. Indeed, by the relation $V_\alpha=\Phi_\alpha v_\alpha \Phi_\alpha^\ast$ and by the use of \eqref{eq: k_ab in p_a and p_b} it follows
\begin{align*}\label{Kern von K}
&\left(\tilde{K}_{\alpha\beta}(z) f\right)(x_\alpha,p_\alpha) 
\\ &= \int_{\mathbb{R}^4} \mathrm{d}k_\alpha \frac{\mathrm{e}^{\mathrm{i} k_\alpha x_\alpha}|v_\alpha|^{\frac{1}{2}}(x_\alpha)}{H^0(k_\alpha,p_\alpha)-z}\int_{\mathbb{R}^4} \mathrm{d}x_\beta \mathrm{e}^{-\mathrm{i}k_\beta x_\beta}|v_\beta(x_\beta)|^{\frac{1}{2}}f(x_\beta,p_\beta)
\\ &= c\int_{\mathbb{R}^4} \mathrm{d}p_\beta \frac{\mathrm{e}^{\mathrm{i}x_\alpha\left( d_{\alpha\beta}p_\alpha+e_{\alpha\beta}p_\beta \right)}|v_\alpha|^{\frac{1}{2}}(x_\alpha)}{H^0_{\alpha\beta}(p_\alpha,p_\beta)-z}\int_{\mathbb{R}^4} \mathrm{d}x_\beta \mathrm{e}^{-\mathrm{i}x_\beta\left( d_{\beta\alpha}p_\alpha+e_{\beta\alpha}p_\beta \right)}|v_\beta(x_\beta)|^{\frac{1}{2}}f(x_\beta,p_\beta),
\end{align*}
i.e. the kernel of $\tilde{K}_{\alpha\beta}(z)$ is of the form
\begin{equation}\label{eq: Kernel K}
\tilde{K}_{\alpha\beta}(z) \left( (x,p),(x',p') \right) = c\mathrm{e}^{\mathrm{i}x p d_{\alpha \beta }}\frac{|v_\alpha(x)|^{\frac{1}{2}}\mathrm{e}^{\mathrm{i}x p' e_{\alpha \beta }}\mathrm{e}^{-\mathrm{i}x' p d_{ \beta \alpha}}|v_\beta(x')|^{\frac{1}{2}}}{H_{\alpha \beta}^0(p,p')-z}\mathrm{e}^{-\mathrm{i}x' p' e_{ \beta \alpha}},
\end{equation}
where the constants $d_{\alpha\beta},e_{\alpha\beta}$ are given by \eqref{eq: k_ab in p_a and p_b}. By the use of \eqref{eq: H_0 etimate2} with $\kappa=\kappa'=\frac{1}{2}$ it follows that $Z_{\alpha\beta}^R(z)$ belongs to the Hilbert-Schmidt class for every $z\leq0$. Using estimate \eqref{eq: H_0 etimate1} one can see that the norm of the operator $Y_{\alpha\beta}^R(z)$ is bounded by $CR^{-2}$ for every $z\leq 0$, where $C$ does not depend on $z$. Hence, we have
\begin{equation*}
\Vert\tilde{K}_{\alpha\beta}(z)-Z_{\alpha\beta}^R(z)\Vert = \Vert Y_{\alpha\beta}^R(z)\Vert \rightarrow 0
\end{equation*}
as $R \rightarrow \infty$, which completes the proof.
\end{proof}
By Lemma \ref{lem: K compact} we have $N(z)<\infty$ for every $z<0$, since $W_\alpha(z)$ is bounded for such $z$. Recall relation \eqref{eq: relation of W and w}
\begin{equation}\label{eq: representation of W and w 2}
W_\alpha(z)= \Phi_\alpha w_\alpha\left( z-\frac{p_\alpha^2}{2n_\alpha} \right) \Phi_\alpha^\ast.
\end{equation} The critical case is the existence of a resonance of the two-body Hamiltonian $h_\alpha$, which affects the behaviour of the operator $w_\alpha(z)$ as $z\rightarrow 0$. 
If $h_\alpha$ has a resonance at zero, then according to Lemma \ref{lem: sqrt w(z)} the operator $w_\alpha^{\frac{1}{2}}(z)$ has the representation
\begin{equation}\label{eq: repres w alpha}
w_\alpha^{\frac{1}{2}}(z)= \frac{\langle \cdot , \varphi \rangle \varphi}{\Vert \varphi \Vert \sqrt{z\left(\ln|z|-\tau_\alpha\right)}} + (z(\ln|z|-\tau_\alpha))^{-\frac{1-\delta}{2}}\tilde{w}_\alpha^{(\delta)}(z),
\end{equation}
where $|z| < 1$ can be chosen sufficiently small, such that $\ln|z|-\tau_\alpha <0$. We only have the representation \eqref{eq: repres w alpha} where $\left|z-\frac{p_\alpha^2}{2n_\alpha}\right|$ is sufficiently small and $\ln\left(z-\frac{p_\alpha^2}{2n_\alpha}\right)-\tau_\alpha<0$. Therefore, for every $\alpha$ we introduce the following auxiliary function $\zeta_\alpha : (-\infty,0) \rightarrow \mathbb{R}$, where $\zeta_\alpha \in C^{\infty}(\mathbb{R}_{-})$, $\zeta_\alpha(t)> 0$ for all $t<0$ and
\begin{equation}\label{eq: definition of zeta}
\zeta_\alpha(t)=
\begin{cases}
\sqrt{t\left(\ln |t| - \tau_\alpha\right)}, & t \in (\mu_\alpha,0)
\\ 1 , & t \leq -1
\end{cases}
\end{equation}
The constant $\mu_\alpha\in(-1,0)$ is chosen such that $\ln\left|t\right|-\tau_\alpha<0$ holds for all $t \in [\mu_\alpha,0)$. This allows us to represent $w_\alpha^{\frac{1}{2}}(z)$ as \eqref{eq: repres w alpha} not only for small $z$ but for every $z<0$ by defining the operator
\begin{equation}\label{eq: def u delta}
\tilde{u}_\alpha^{(\delta)}(z)=
\begin{cases}
\tilde{w}_\alpha^{(\delta)}(z), &z\in(\mu_\alpha,0)
\\ \zeta_\alpha(z)^{-\delta}\left(\zeta_\alpha(z)w_\alpha^{\frac{1}{2}}(z)-\Vert \varphi \Vert^{-1}\langle\cdot,\varphi\rangle\varphi\right), & z\in(-\infty,\mu_\alpha]
\end{cases}
\end{equation}
Since $w_\alpha^{\frac{1}{2}}(z)$ is uniformly bounded for $z\leq \mu_\alpha<0$ and $\tilde{w}_\alpha^{(\delta)}(z)$ is continuous up to $z=0$, it follows that
\begin{equation}
w_\alpha^\frac{1}{2}(z)=\zeta_\alpha(z)^{-1}\Vert \varphi \Vert^{-1}\langle\cdot ,\varphi \rangle \varphi + \zeta_\alpha(z)^{-1+\delta}\tilde{u}_\alpha^{(\delta)}(z)
\end{equation}
holds true for every $z<0$ and the operator $\tilde{u}_\alpha^{(\delta)}(z)$ is continuous up to $z=0$. By relation \eqref{eq: relation of W and w} it is evident that for $z = 0$ the kernel of $A_{\alpha\beta}(z)$ admits a singularity in $p_\alpha=0$ and $p_\beta=0$. In the following we will decompose the kernel of $A_{\alpha\beta}(z)$ into four kernels. Simply put, we will cut the region in the variables $p_\alpha,p_\beta$ where both $|p_\alpha|,|p_\beta|$ are small, both $|p_\alpha|,|p_\beta|$ are large and the other two cases where $|p_\alpha|$ is small and $|p_\beta|$ is large, and vice versa. 

Here, it should be noted that in dimension four the mixed cases of one of the variables $|p_\alpha|$,$|p_\beta|$ being small and the other one being large is more complicated compared to the three-dimensional case \cite{Sobolev}. After squaring the kernel the resolvent provides in both cases a decay like $|p_\alpha|^{-4}$ and $|p_\beta|^{-4}$, which in dimension three yields the Hilbert-Schmidt property. This argument cannot be adapted to the four-dimensional case.
\begin{lem}\label{lem: Lemma for main result}
Let $\Gamma_\alpha(z)$ be the operator of multiplication by $\zeta_\alpha\left(z-\frac{p_\alpha^2}{2n_\alpha} \right)$, i.e.
\begin{equation}
(\Gamma_\alpha(z) f)(k_\alpha,p_\alpha)=\zeta_\alpha\left(z-\frac{p_\alpha^2}{2n_\alpha} \right)\cdot f(k_\alpha,p_\alpha).
\end{equation}
Then the operator
\begin{equation}
M_{\alpha\beta}(z)=\Gamma_\alpha(z)^{-1}K_{\alpha\beta}(z)\Gamma_\beta(z)^{-1}
\end{equation}
is compact for every $z\leq 0$.
\end{lem}
\begin{proof}
We consider the operator
\begin{equation*}
\tilde{M}_{\alpha\beta}(z)=\Phi_\alpha^\ast M_{\alpha\beta}(z)\Phi_\beta .
\end{equation*}
The compactness for $z<0$ follows from Lemma \ref{lem: K compact}. We only need to consider the case $z=0$. Similar to \eqref{eq: Kernel K} the kernel of $\tilde{M}_{\alpha\beta}(0)$ is given by
\begin{equation*}
\tilde{M}_{\alpha\beta}\left( (x,p),(x',p') \right) = c\mathrm{e}^{\mathrm{i}x p d_{\alpha \beta }}\frac{|v_\alpha(x)|^{\frac{1}{2}}\mathrm{e}^{\mathrm{i}x p' e_{\alpha \beta }}\mathrm{e}^{-\mathrm{i}x' p d_{ \beta \alpha}}|v_\beta(x')|^{\frac{1}{2}}}{\zeta_\alpha\left(-\frac{p^2}{2n_\alpha} \right) H_{\alpha \beta}^0(p,p')\zeta_\beta\left(-\frac{p'^2}{2n_\beta} \right)}\mathrm{e}^{-\mathrm{i}x' p' e_{ \beta \alpha}}.
\end{equation*}
Let $\mu_\alpha,\mu_\beta<0$ be in accordance with \eqref{eq: definition of zeta} and $0<r<\min(|\mu_\alpha|,|\mu_\beta|\}$. Denote by $\chi_r(p)$ the multiplication by the characteristic function of $\left\{p\in \mathbb{R}^4 \ : \frac{p^2}{2n}< r \right\}$, where $n=\min\{n_\alpha,n_\beta\}$. We decompose 
\begin{equation*}
\tilde{M}_{\alpha\beta} = \tilde{M}_{\alpha\beta}^1+\tilde{M}_{\alpha\beta}^2+\tilde{M}_{\alpha\beta}^3,
\end{equation*}
where
\begin{align}
\tilde{M}_{\alpha\beta}^1&=\chi_r(p)\tilde{M}_{\alpha\beta}\chi_r(p'), \label{eq: kernel M1}
\\ \tilde{M}_{\alpha\beta}^2&=\chi_r(p)\tilde{M}_{\alpha\beta}(I-\chi_r(p'))+(I-\chi_r(p))\tilde{M}_{\alpha\beta}\chi_r(p'), \label{eq: kernel M2}
\\ \tilde{M}_{\alpha\beta}^3&=(I-\chi_r(p))\tilde{M}_{\alpha\beta}(I-\chi_r(p')).
\end{align}
The compactness of $\tilde{M}_{\alpha\beta}^3$ follows from Lemma \ref{lem: K compact}. 
\\Let us prove that $\tilde{M}_{\alpha\beta}^2$ is compact. We consider only the first summand, the second one can be treated analogously. Let $R > r>0$ be fixed. Then the first summand of $\tilde{M}_{\alpha\beta}^2$ can be written as
\begin{equation} 
\chi_r(p)\tilde{M}_{\alpha\beta}(I-\chi_r(p'))=X_{\alpha\beta}+Y_{\alpha\beta},
\end{equation}
where 
\begin{equation}
X_{\alpha\beta} = \chi_r(p)\tilde{M}_{\alpha\beta}(\chi_R(p')-\chi_r(p')), \qquad \qquad Y_{\alpha\beta} = \chi_r(p)\tilde{M}_{\alpha\beta}(I-\chi_R(p')).
\end{equation}
By the use of $H_{\alpha\beta}^0(p,p') \geq cp'^2$, the absolute value of the kernel of $X_{\alpha\beta}$ can be estimated from above by
\begin{equation*}
c\chi_r(p)\frac{|v_\alpha(x)|^{\frac{1}{2}}|v_\beta(x')|^{\frac{1}{2}}}{|p||p'|^2}(\chi_R(p')-\chi_r(p')),
\end{equation*}
which is square-integrable with respect to the arguments $x,x',p,p'$. The kernel of $Y_{\alpha\beta}$ is given by
\begin{equation*}
Y_{\alpha\beta}((x,p),(x',p'))= c\chi_r(p)(1-\chi_R(p'))\mathrm{e}^{\mathrm{i}x p d_{\alpha \beta }}\frac{|v_\alpha(x)|^{\frac{1}{2}}\mathrm{e}^{\mathrm{i}x p' e_{\alpha \beta }}\mathrm{e}^{-\mathrm{i}x' p d_{ \beta \alpha}}|v_\beta(x')|^{\frac{1}{2}}}{\zeta_\alpha\left(-\frac{p^2}{2n_\alpha} \right) H_{\alpha \beta}^0(p,p')}\mathrm{e}^{-\mathrm{i}x' p' e_{ \beta \alpha}}.
\end{equation*}
We will show that $Y_{\alpha\beta}^\ast Y_{\alpha\beta}$ is continuous and the operator norm tends to zero as $R\rightarrow \infty$. The kernel of $Y_{\alpha\beta}^\ast Y_{\alpha\beta}$ is given by
\begin{align*}
Y_{\alpha\beta}^\ast &Y_{\alpha\beta}\left((x'',p''),(x',p')\right) = \int \int \overline{Y_{\alpha\beta}\left((x,p),(x',p')\right)}Y_{\alpha\beta}\left((x,p),(x'',p'')\right) \mathrm{d}x\mathrm{d}p.
\end{align*}
Hence, we have
\begin{align}\label{eq: estimate YY*}
|Y_{\alpha\beta}^\ast Y_{\alpha\beta}\left((x'',p''),(x',p')\right)| 
\leq c|\widehat{v_\alpha}(p'-p'')||v_\beta(x')|^{\frac{1}{2}}|v_\beta(x'')|^{\frac{1}{2}}J(p',p'')(I-\chi_R(p''))(1-\chi_R(p')),
\end{align}
where
\begin{equation*}
\widehat{v_\alpha}(p'-p'') = \int_{\mathbb{R}^4} |v_\alpha(x)| \mathrm{e}^{-\mathrm{i}e_{\alpha\beta}x(p'-p'')} \, \mathrm{d}x, \qquad J(p',p'')=\int_{\left\{|p|<\sqrt{2nr}\right\}} \frac{1}{p^2(p^2+p'^2)(p^2+p''^2)}\, \mathrm{d}p.
\end{equation*}
In view of the characteristic functions $(I-\chi_R(p''))$ and $(1-\chi_R(p'))$ we have $|p'|,|p''|\geq c>0$ and therefore
\begin{equation*}
J(p',p'')\leq\frac{C}{p'^2p''^2}
\end{equation*}
for such $p'$ and $p''$, which implies
\begin{equation}\label{eq: estimate kernel Yab}
|Y_{\alpha\beta}^\ast Y_{\alpha\beta}\left((x'',p''),(x',p')\right)| \leq C\frac{\widehat{|v_\alpha}(p'-p'')|}{p'^2p''^2}|v_\beta(x')|^{\frac{1}{2}}|v_\beta(x'')|^{\frac{1}{2}}(I-\chi_R(p''))(1-\chi_R(p')).
\end{equation}
For $\xi_1,\xi_2 \in \mathbb{R}^4\backslash\{0\}$ and $b>4$ we define the function
\begin{equation*}
y(\xi_1,\xi_2)=(1+|\xi_1|)^{-\frac{b}{2}}|\xi_2|^{-2}.
\end{equation*}
By assumption \eqref{eq: assumptions on potentials in dimension 4} we have $v_\alpha \in L^1(\mathbb{R}^4) \cap L^\infty(\mathbb{R}^4)$ and $\widehat{v_\alpha} \in L^2(\mathbb{R}^4) \cap L^\infty(\mathbb{R}^4)$. Hence, by the use of \eqref{eq: estimate kernel Yab} and $(1+|\cdot|)^{-\frac{b}{2}}|v_\beta(\cdot)|^\frac{1}{2}\in L^1(\mathbb{R}^4)$ we have 
\begin{equation*}
\int \left|Y_{\alpha\beta}^\ast Y_{\alpha\beta}\left((x'',p''),(x',p')\right)y(x',p')\right|\, \mathrm{d}x' \mathrm{d}p' \leq \frac{|v_\beta(x'')|^\frac{1}{2}}{|p''|^2}C_R  \leq y(x'',p'') C_R,
\end{equation*}
where $C_R\rightarrow 0$ as $R\rightarrow \infty$. By symmetry we also have
\begin{equation*}
\int \left|Y_{\alpha\beta}^\ast Y_{\alpha\beta}\left((x'',p''),(x',p')\right)y(x'',p'')\right|\, \mathrm{d}x'' \mathrm{d}p'' \leq  \frac{|v_\beta(x')|^\frac{1}{2}}{|p'|^2}C_R \leq y(x',p')C_R.
\end{equation*}
Hence, we can apply the Schur test (see \cite{Schur}) to conclude that $Y_{\alpha\beta}$ is a bounded operator on $L^2(\mathbb{R}^4)$, where the operator norm tends to zero as $R\rightarrow \infty$. By applying the same arguments to the second kernel of \eqref{eq: kernel M2} we conclude that $\tilde{M}_{\alpha\beta}^2$ is compact.

It remains to show that $\tilde{M}_{\alpha\beta}^1$ is compact. By definition of the function \eqref{eq: definition of zeta} and in view of the characteristic functions $\chi_r(p), \chi_r(p')$, where $0<r<\mu<\min\{|\mu_\alpha|,|\mu_\beta|\}$, it is sufficient to show that the integral
\begin{equation}
\int_{|p|<\mu} \int_{|p'|<\mu} K\left(p,p' \right)\mathrm{d}p'\mathrm{d}p
\end{equation}
is finite, where $\mu>0$ is sufficiently small and the kernel $K$ is given by
\begin{equation*}
K(p,p') = \frac{1}{|p|^2\left|\ln|p|\right|\left(H_{\alpha\beta}^0(p,p')\right)^2|p'|^2\left|\ln|p'|\right|}.
\end{equation*}
Note that
\begin{equation*}
\left(H_{\alpha\beta}^0(p,p')\right)^2 \geq c|p|^{4\kappa}|p'|^{4\kappa'}, \qquad \qquad \kappa+\kappa'=1.
\end{equation*}
We set $\kappa=0$ and use spherical coordinates $p=(\omega,\rho),\ p'=(\omega',\rho')$ to obtain
\begin{align}
\int\limits_{|p|<\mu} \int\limits_{|p'|<\mu}K(p,p')\, \mathrm{d}p'\mathrm{d}p & = \int\limits_{|p|<\mu} \left( \int\limits_{|p'|\leq |p|} K(p,p')\, \mathrm{d}p'+\int\limits_{|p|<|p'|<\mu}K(p,p')\, \mathrm{d}p' \right)\mathrm{d}p  \notag
\\ &\leq C\int\limits_{|p|<\mu}   \frac{1}{p^2 \left|\ln|p|\right|} \left(\int\limits_{|p|\leq|p'|<\mu} \frac{1}{p'^6\left|\ln|p'|\right|}\, \mathrm{d}p'\right) \mathrm{d}p \notag
\\ & \leq C' \int_0^\mu \frac{\rho}{|\ln\rho|}\left( \int_{\rho}^\mu  \frac{1}{\rho'^3\left|\ln \rho' \right|}\mathrm{d}\rho'\right) \mathrm{d}\rho \notag
\\ &= C' \int_0^\mu  \frac{\rho}{\left|\ln \rho\right|} F(\mu,\rho)\, \mathrm{d}\rho, \label{eq: int for K(p,p prime)}
\end{align}
where the function $F$ is given by
\begin{equation*}
F(\mu,\rho)=\int_{\rho}^\mu  \frac{1}{\rho'^3\left|\ln \rho' \right|}\mathrm{d}\rho'=- \frac{1}{2\rho'^2 \left| \ln \rho' \right|}\mathop{\bigg|}\limits_{\rho}^\mu - \int_{\rho}^\mu \frac{1}{2\rho'^3\left| \ln \rho' \right|^2}.
\end{equation*}
For $\mu>0$ sufficiently small we have $\left| \ln \rho' \right| \geq 1$ and therefore
\begin{equation}\label{eq: integration by parts for F(mu,rho)}
\left| F(\mu,\rho) \right| \leq C(\mu) + \frac{1}{2\rho^2 \left| \ln  \rho \right|} + \frac{1}{2}\left| F(\mu,\rho) \right|.
\end{equation}
Hence, by inserting \eqref{eq: integration by parts for F(mu,rho)} into \eqref{eq: int for K(p,p prime)} we get
\begin{equation}
\int\limits_{|p|<\mu} \int\limits_{|p'|<\mu}K(p,p')\, \mathrm{d}p'\mathrm{d}p \leq C' \int_0^\mu  \frac{\rho}{\left|\ln \rho\right|} F(\mu,\rho)\, \mathrm{d}\rho \leq C_1+C_2\int_{0}^\mu \frac{1}{\rho(\ln\rho)^2}\, \mathrm{d}\rho<\infty,
\end{equation}
which completes the proof.
\end{proof}
\section{Proof of the main results}
We are ready to prove Theorem \ref{main result} and Theorem \ref{main result 2}.
\begin{proof}[Proof of Theorem \ref{main result}]
Let $d=4$. At first we assume that every two-body Hamiltonian $h_\alpha, \ \alpha \in \{12,23,31\}$ has a virtual level at zero. According to Lemma \ref{lem: <v,f>=0}, $\lambda=0$ is not an eigenvalue of $h_\alpha$. In the case of resonances in every subsystem, every entry $A_{\alpha\beta}(z)$ of $A(z)$ can be represented as
\begin{align*}
A_{\alpha\beta}(z)=&\Pi_\alpha \left( \Gamma_\alpha(z) \right)^{-1}K_{\alpha\beta}(z)\Gamma_\beta(z)^{-1}\Pi_\beta + \tilde{U}_\alpha^{(\delta)}(z)\Gamma_\alpha(z)^{-1+\delta}K_{\alpha\beta}(z)\left(\Gamma_\beta(z) \right)^{-1}\Pi_\beta
\\ & + \Pi_\alpha\Gamma_\alpha (z)^{-1}K_{\alpha\beta}(z)\left(\Gamma_\beta(z) \right)^{-1+\delta}\tilde{U}_\beta^{(\delta)}(z) +\tilde{U}_\alpha^{(\delta)}(z) \Gamma_\alpha(z)^{-1+\delta}K_{\alpha\beta}(z)\Gamma_\beta(z)^{-1+\delta}\tilde{U}_\beta^{(\delta)}(z),
\end{align*}
where the operator $\Pi_\alpha$ is defined by
\begin{equation*}
(\Pi_\alpha f)(k_\alpha,p_\alpha)= \Vert \varphi_\alpha \Vert^{-1}(\Phi_\alpha \varphi_\alpha)(k_\alpha)\int f(k_\alpha',p_\alpha)\overline{(\Phi_\alpha\varphi_\alpha)(k_\alpha')}\, \mathrm{d}k_\alpha'
\end{equation*}
and $\tilde{U}_\alpha^{(\delta)}(z)$ is given by
\begin{equation*}
\tilde{U}_\alpha^{(\delta)}(z)=\Phi_\alpha \tilde{u}_\alpha^{(\delta)}\left(z-\frac{p_\alpha^2}{2n_\alpha}\right) \Phi_\alpha^\ast,
\end{equation*}
where $\tilde{u}_\alpha^{(\delta)}(z)$ is defined by \eqref{eq: def u delta}. Since $\Pi_\alpha,\Pi_\beta,\tilde{U}_\alpha^{(\delta)}(z),\tilde{U}_\beta^{(\delta)}(z)$ are bounded operators for $z\leq 0$, the finiteness of $\sigma_{\mathrm{disc}}(H)$ follows from Lemma \ref{lem: Lemma for main result} and Theorem \ref{thm counting functions of A and N}. Now assume that one subsystem, say $\alpha$, does not have a resonance. In case of $\lambda=0$ being a regular point of $h_\alpha$, the operator $w_\alpha(z)$ is continuous up to $z=0$. Indeed, one can easily see that $\mu =1 $ is not an eigenvalue of the operator with the kernel 
\begin{equation*}
G_\alpha(x,y)=\frac{m_\alpha}{2\pi^2}\frac{|v_\alpha|^\frac{1}{2}(x)|v_\alpha|^\frac{1}{2}(y)}{|x-y|^2}.
\end{equation*}
Similar to Lemma \ref{lem: expansion of w} one has
\begin{equation*}
w_\alpha(z)= \left( I-G_\alpha+o(1) \right)^{-1} = \left(I-G_\alpha\right)^{-1}+o(1), \ z \rightarrow 0.
\end{equation*}
This implies the finiteness of $\sigma_{\mathrm{disc}}(H)$ in this case aswell. If $\lambda=0$ is an eigenvalue of $h_\alpha$, then we do not need to distinguish between dimensions $d=4$ and $d\geq5$, since the finiteness of $\sigma_{\mathrm{disc}}(H)$ follows from similar arguments as in \cite{Zhislin}. Indeed, by Lemma \ref{lem: behaviour of the resonance state} the virtual level is always an eigenvalue for $d \geq 5$. By Corollary \ref{Cor: proof} there exists a constant $\mu>0$, such that for every function $g \in \dot{H}^1(\mathbb{R}^d)\backslash\{0\}$ with $\langle \nabla g,\nabla f\rangle=0$ one has
\begin{equation*}
\langle (-\Delta+v)g,g\rangle \geq \mu \Vert |\nabla g| \Vert^2.
\end{equation*}
Now we can repeat the same arguments as in the proof of \cite[Theorem 2.1]{Zhislin} to prove the existence of a finite-dimensional subspace $M\subset \mathcal{D}(H)$, such that
\begin{equation*}
\langle H\psi,\psi\rangle \geq 0
\end{equation*}
holds for every $\psi \perp M$. This concludes the proof.
\end{proof}
The absence of the Efimov effect for the antisymmetric case follows from Theorem \ref{main result} by a slight modification.
\begin{proof}[Proof of Theorem \ref{main result 2}] 
We have $\langle v_\alpha,\varphi\rangle=0$ for every antisymmetric function $\varphi$, since the potentials $v_\alpha$ satisfy $v_\alpha(x_{ij})=v_\alpha(-x_{ij})$. Hence, by Lemma \ref{lem: <v,f>=0} a virtual level of $h_\alpha^{\mathrm{as}}$ is always an eigenvalue for $d\geq 4$. It is easy to see that this eigenvalue has always finite multiplicity. Let $E_d$ be the corresponding eigenspace. Similar to Corollary \ref{Cor: proof}, by considering the orthogonal complement of $E_d$ with respect to $\Vert \cdot \Vert_{(1)}$, there exists a constant $\mu_d>0$ such that for every function $g\perp E_d$ in $\dot{H}^1(\mathbb{R}^d)$ one has
\begin{equation*}
\langle h_\alpha^{\mathrm{as}} g,g\rangle \geq \mu_d \Vert |\nabla g| \Vert^2.
\end{equation*}
The finiteness of $\sigma_{\mathrm{disc}}(H^{\mathrm{as}})$ now follows from \cite[Theorem 2.1]{Zhislin}.
\end{proof}

\begin{appendices}
\section{Coordinate system}
Let $x_i$ be the coordinate of the particle with mass $m_i,\ i \in \{1,2,3\}$ and denote by $k_i$ the conjugate variable of $x_i$. Let $(x_\alpha,y_\alpha), \ \alpha \in \{12,23,31\}$ be any pair of Jacobi coordinates and introduce the set of variables conjugate with respect to the Jacobi-coordinates
\begin{align}
&k_{12}= \frac{m_2k_1-m_1k_2}{m_1+m_2}, \qquad p_{12}= \frac{m_3(k_1+k_2)-(m_1+m_2)k_3}{m_1+m_2+m_3}, \label{Jacobi in momentum 1}
\\ &k_{23}= \frac{m_3k_2-m_2k_3}{m_2+m_3}, \qquad p_{23}= \frac{m_1(k_2+k_3)-(m_2+m_3)k_1}{m_1+m_2+m_3},
\\ &k_{31}= \frac{m_1k_3-m_3k_1}{m_3+m_1}, \qquad p_{31}= \frac{m_2(k_3+k_1)-(m_3+m_1)k_2}{m_1+m_2+m_3}. \label{Jacobi in momentum 3}
\end{align}
The reduced masses $m_\alpha, n_\alpha$ in \eqref{eq: H^0 in all coordinates} are given by 
\begin{align}
&m_{12}= \frac{m_1m_2}{m_1+m_2}, \qquad n_{12}= \frac{m_3(m_1+m_2)}{m_1+m_2+m_3}, \label{eq: reduced masses1}
\\ &m_{23}= \frac{m_2m_3}{m_2+m_3}, \qquad n_{23}= \frac{m_1(m_2+m_3)}{m_2+m_2+m_3}, 
\\ &m_{31}= \frac{m_3m_1}{m_3+m_1}, \qquad n_{31}= \frac{m_2(m_3+m_1)}{m_2+m_2+m_3}. \label{eq: reduced masses2}
\end{align}
The relations \eqref{eq: k_ab in p_a and p_b} of coordinates $(p_\alpha,p_\beta)$ and $(k_\alpha,p_\alpha)$ are given by
\begin{align}
&k_{12}=-p_{23}-\frac{m_1}{m_1+m_2}p_{12} = p_{31}+\frac{m_2}{m_1+m_2}p_{12}, \label{Jacobi with p and k 1}
\\ &k_{23}=-p_{31}-\frac{m_2}{m_2+m_3}p_{23} = p_{12}+\frac{m_3}{m_1+m_2}p_{23},
\\ &k_{31}=-p_{12}-\frac{m_3}{m_1+m_2}p_{31} = p_{23}+\frac{m_1}{m_1+m_2}p_{31}.\label{Jacobi with p and k 3}
\end{align}
The kinetic energy $H^0$ can now be expressed by any pair $(p_\alpha,p_\beta)$, i.e. $H^0$ takes the form
\begin{equation}\label{AP last}
\frac{p_{12}^2}{2m_{{23}}}+ \frac{\langle p_{12},p_{23} \rangle}{m_2}+\frac{p_{23}^2}{2m_{{12}}} =\frac{p_{23}^2}{2m_{{31}}}+ \frac{\langle p_{23},p_{31} \rangle}{m_3}+\frac{p_{31}^2}{2m_{{23}}}=\frac{p_{31}^2}{2m_{{12}}}+ \frac{\langle p_{31},p_{12} \rangle}{m_1}+\frac{p_{12}^2}{2m_{{31}}}.
\end{equation}
For more details see \cite{Fad}.
\section{Faddeev equations}
We only sketch a brief derivation of the Faddeev equations (see \cite{Motovilov} for more details).
\\ Consider the bound-state equation for the three-body Hamiltonian $H$.
\begin{equation}\label{eq: eigenvalue three body}
\left( H_0 +\sum_{\alpha}V_\alpha\right)u=zu,
\end{equation}
where $z<0$ and $\alpha\in\{12,23,31\}$. Denote $R_0(z)=(H_0-z)^{-1}$, then
\begin{equation*}
u=-R_0(z)\sum_{\alpha}V_\alpha u.
\end{equation*}
By decomposing $u$ in the so called Faddeev components
\begin{equation*}
u=\sum_\alpha u_\alpha, \qquad u_\alpha = -R_0V_\alpha u
\end{equation*}
and defining
\begin{equation*}
H_\alpha=H_0+V_\alpha \qquad \text{and} \qquad R_\alpha(z)=(H_\alpha-z)^{-1},
\end{equation*}
one has
\begin{align*}
u_\alpha = -R_0(z)V_\alpha\sum_\alpha u_\alpha &\Longleftrightarrow u_\alpha +R_0(z)V_\alpha u_\alpha = -R_0(z)V_\alpha(u_\beta+u_\gamma)
\\ &\Longleftrightarrow R_\alpha(z)(H_\alpha-z)(I+R_0(z)V_\alpha)u_\alpha = -R_0(z)V_\alpha(u_\beta+u_\gamma)
\\ &\Longleftrightarrow R_\alpha(z)(H_0+V_\alpha-z)u_\alpha = -R_0(z)V_\alpha(u_\beta+u_\gamma)
\\ &\Longleftrightarrow u_\alpha = -R_\alpha(z)V_\alpha(u_\beta+u_\gamma).
\end{align*}
By the assumption $V_\alpha\leq 0$ and the resolvent identity
\begin{equation*}
R_\alpha(z)=R_0(z)-R_0(z)V_\alpha R_\alpha(z),
\end{equation*}
one arrives at
\begin{equation}\label{eq: faddeev components}
u_\alpha = R_0(z)(I+|V_\alpha| R_\alpha)|V_\alpha|(u_\beta+u_\gamma).
\end{equation}
Now we define
\begin{equation*}
W_\alpha(z)=I+|V_\alpha|^{\frac{1}{2}} R_\alpha(z)|V_\alpha|^{\frac{1}{2}},
\end{equation*}
then equation \eqref{eq: faddeev components} becomes
\begin{equation*}
u_\alpha = R_0(z)|V_\alpha|^\frac{1}{2}W_\alpha(z)|V_\alpha|^\frac{1}{2}(u_\beta+u_\gamma).
\end{equation*}
By making the substitution
\begin{align*}
&f_\alpha=W_\alpha^\frac{1}{2}(z)|V_\alpha|^\frac{1}{2}(u_\beta+u_\gamma),
\\ &f_\beta=W_\beta^\frac{1}{2}(z)|V_\beta|^\frac{1}{2}(u_\alpha+u_\gamma),
\\ &f_\gamma=W_\gamma^\frac{1}{2}(z)|V_\gamma|^\frac{1}{2}(u_\alpha+u_\gamma),
\end{align*}
the system of equations is now given by the Faddeev equations
\begin{align}
f_{12}&=W_{12}^\frac{1}{2}(z)|V_{12}|^\frac{1}{2}R_0(z)|V_{23}|^\frac{1}{2}W_{23}^\frac{1}{2}(z)f_{23}+W_{12}^\frac{1}{2}(z)|V_{12}|^\frac{1}{2}R_0(z)|V_{31}|^\frac{1}{2}W_{31}^\frac{1}{2}(z)f_{31},
\\ f_{23}&=W_{23}^\frac{1}{2}(z)|V_{23}|^\frac{1}{2}R_0(z)|V_{12}|^\frac{1}{2}W_{12}^\frac{1}{2}(z)f_{12}+W_{23}^\frac{1}{2}(z)|V_{23}|^\frac{1}{2}R_0(z)|V_{31}|^\frac{1}{2}W_{31}^\frac{1}{2}(z)f_{31},
\\ f_{31}&=W_{31}^\frac{1}{2}(z)|V_{31}|^\frac{1}{2}R_0(z)|V_{12}|^\frac{1}{2}W_{12}^\frac{1}{2}(z)f_{12}+W_{31}^\frac{1}{2}(z)|V_{31}|^\frac{1}{2}R_0(z)|V_{23}|^\frac{1}{2}W_{23}^\frac{1}{2}(z)f_{23}.
\end{align}
In other words, the eigenvalue equation \eqref{eq: eigenvalue three body} is now formulated as
\begin{equation*}
A(z)F=F, \qquad F=(f_{12},f_{23},f_{31}),
\end{equation*}
where $A(z)$ is a $3\times 3-$matrix with entries
\begin{equation*}
A_{\alpha\beta}(z) = W_\alpha^\frac{1}{2}(z) |V_\alpha|^\frac{1}{2}R_0(z)|V_\beta|^\frac{1}{2} W_\beta^\frac{1}{2}(z).
\end{equation*}
\end{appendices}
\section*{Acknowledgements}
We are deeply grateful to Timo Weidl and Semjon Vugalter for proposing the problem, helpfull discussions and a wide range of suggestions. We would also like to thank Jens Wirth for many usefull hints.
The main part of the research for this paper was done while visiting the Mittag-Leffler Institute within the semester program \textit{Spectral Methods in Mathematical Physics}.
\bibliography{Bib}{}
\bibliographystyle{abbrv}
\end{document}